\documentclass{article}

\usepackage{arxiv}

\usepackage[utf8]{inputenc} 
\usepackage[T1]{fontenc}    
\usepackage{hyperref}       
\usepackage{url}            
\usepackage{booktabs}       
\usepackage{amsfonts}       
\usepackage{nicefrac}       
\usepackage{microtype}      
\usepackage{lipsum,enumerate,bm,mathrsfs}
\usepackage{graphicx,amsmath,multicol,multirow}
\usepackage{url}
\usepackage{amsthm}
\usepackage[style=apa, backend=biber, natbib=true]{biblatex}
\usepackage{adjustbox,xr,comment} 

\usepackage{amsmath}
\usepackage{amssymb}
\usepackage{amsthm,tikz}
\usepackage{colortbl,caption,xr}
\usepackage{adjustbox}  
\usepackage{graphicx}
\usepackage{lscape}
\usepackage{rotating}
\usepackage{graphicx}
\usepackage{float}
\usepackage{url}
\usepackage{subfigure}
\usepackage{dcolumn}
\usepackage{multirow}
\usepackage{verbatim,booktabs, array,arydshln}
\usepackage{color}
\usepackage{setspace}
\usepackage{bbm}
\usepackage{booktabs}
\usepackage{hhline}
\usepackage{bm}
\usepackage{dsfont,import}
\usepackage{enumerate}
\usepackage{threeparttable}
\usepackage{array}
\usepackage{caption}
\usepackage{tikz}
\usetikzlibrary{shapes.geometric, arrows, shadows}
\usepackage{graphicx}
\usetikzlibrary{shapes.geometric, arrows, shadows,positioning,arrows.meta,calc}
\usepackage{graphicx}
\usepackage{booktabs,multicol,multirow,tabularx,array}
\usepackage{longtable}
\usepackage{siunitx,array}
\usepackage{caption}
\usepackage{graphicx}

\usepackage{siunitx}

\newtheorem{theorem}{Theorem}[section]   

\theoremstyle{definition}

\theoremstyle{remark}

\theoremstyle{definition}

\theoremstyle{plain}

\tikzstyle{state} = [ellipse, draw, fill=gray!20, text centered, text width=1.5cm, minimum height=1.2cm, drop shadow]
\tikzstyle{arrow} = [thick,->,>=stealth]

\title{Statistical inference with win statistics in cluster-randomized trials with hierarchical composite outcomes}

\author{
 Xi Fang \\
    Department of Biostatistics \\
    Medical College of Wisconsin \\
    Milwaukee, WI, USA\\
\And
Guangyu Tong \\
 Department of Biostatistics \\
    Yale School of Public Health \\
    New Haven, CT, USA\\
\And
Yuan Huang\\
 Department of Biostatistics \\
    Yale School of Public Health \\
    New Haven, CT, USA\\
\And
F. Perry Wilson\\
Section of Nephrology\\
Department of Internal Medicine \\
Yale School of Medicine \\
New Haven, CT, USA\\
\And
Patrick J. Heagerty\\
Department of Biostatistics \\
University of Washington \\
Seattle, WA, USA\\
\And
 Fan Li \\
Department of Biostatistics \\
Yale School of Public Health \\
New Haven, CT, USA\\
  \texttt{fan.f.li@yale.edu} \\
}

\begin{document}
\maketitle
\begin{abstract}
Win statistics have become increasingly popular for analyzing hierarchical composite endpoints in clinical trials. The win ratio, win odds, net benefit, and desirability of outcome ranking (DOOR) share a pairwise-comparison framework and provide complementary summaries of treatment benefit. Despite recent progress in individually randomized trials, statistical inference for these measures in cluster-randomized trials (CRTs) remains underdeveloped. We provide a unified development and comparison of six testing procedures for all four win measures in parallel-arm CRTs: three Wald tests, two randomization-based tests, and a jackknife empirical likelihood ratio test. Through simulations with hierarchical semi-competing risks outcomes, we evaluate type I error rate and power across varying design parameters. With 20 clusters, the randomization-based procedures provided the most stable type I error rate control, while the clustered rank-sum Wald test with a \(t\)-distribution was the most reliable among the Wald tests, occasionally carrying a conservative test size. For the win ratio and win odds, the permutation test also yielded higher power than the clustered rank-sum Wald test. With 100 clusters, differences in type I error rate and power were small. These findings favor permutation inference in small CRTs, with the clustered rank-sum Wald test providing an analytic alternative, particularly for net benefit and DOOR. On the other hand, the analytic Wald procedures offer a computationally convenient choice for CRTs with a large number of clusters. We illustrate the methods by reanalyzing the STRIDE trial and implement all procedures in the \texttt{WinsCRT} R package.
\end{abstract}

\keywords{cluster-randomized trial; jackknife; permutation test; win ratio; U-statistics; type I error rate}

\section{Introduction}\label{intro}

Cluster-randomized trials (CRTs) are increasingly used to evaluate interventions implemented at the level of clinics, hospitals, schools, and communities, particularly in pragmatic and implementation research, where individual randomization may be infeasible or likely to lead to contamination \citep{hayes2017cluster}. In many CRTs, the intervention is expected to affect multiple dimensions of health, leading investigators to collect several clinically relevant outcomes that jointly characterize its benefits and harms. When individual outcome events are uncommon, analyzing each component separately may provide limited statistical precision and an incomplete picture of the intervention’s overall effect. A prespecified composite endpoint can combine information across components, provide a global assessment of benefit and harm, and reduce the multiplicity burden that would arise if each component were evaluated as a separate primary endpoint \citep{freemantle2003composite,huque2011addressing,mao2021statistical,guidance2017multiple}. Conventional composite endpoints, however, may treat events with substantially different clinical implications, such as death and hospitalization, as contributing equally to the treatment comparison. This limitation has motivated growing interest in hierarchical composite endpoints, which incorporate clinical knowledge by ranking outcome components according to their prespecified importance \citep{pocock2012win}.

Composite endpoints also introduce important dependence-related challenges. Even in an individually randomized trial, multiple outcome components recorded for the same individual may be correlated because they share individual-level risk factors and because the occurrence of one event may alter the subsequent risk or observability of another \citep{sankoh2014use,mou2025correlation}. CRTs introduce an additional layer of dependence because outcome vectors from different individuals within the same cluster may be correlated through shared environments, common care processes, or other cluster-level determinants \citep{murray2008design}. Valid hypothesis testing procedures for hierarchical composite endpoints in CRTs must therefore account simultaneously for dependence among component outcomes within individuals as well as dependence of the same outcome across individuals within clusters.

An increasingly popular approach to study composite outcomes is to base inference on win statistics defined through sequentially hierarchical outcomes using pairwise comparisons \citep{pocock2012win,dong2020win}.  Under this prioritized or hierarchical outcome structure, each pair consisting of one treated individual and one control individual is evaluated beginning with the most clinically important component, and the comparison proceeds to the next component in the priority order only when the higher-priority component results in a tie. If no component determines a winner, the pair is considered a tie overall. Thus, lower priority outcomes do not mask more important ones simply because they occur earlier, such as when hospitalization is considered together with death. Because the construction of win statistics depends only on relative ordering, it is invariant to monotone transformations and applies broadly to outcomes that can be ranked, which is particularly appealing for composite endpoints whose components differ in distributional form or scale \citep{zou2023parametric,buyse2010generalized}. The pairwise-comparison framework has motivated a family of win-based summary measures, including the win ratio \citep{pocock2012win}, win odds \citep{dong2020win}, net benefit \citep{buyse2010generalized}, and desirability of outcome ranking (DOOR) \citep{evans2015desirability,barnhart2025sample}, which provide complementary relative, absolute, and probabilistic assessments of treatment benefit under a common comparison rule. These win summary measures are all constructed based on the same underlying pairwise comparison principles, leverage the same win/loss/tie proportions as key ingredients, and the corresponding win statistics can be used to test the null hypothesis of no treatment effect in randomized clinical trials \citep{dong2020win}.

The statistical theory for win statistics has been developed in the independent data setting. For example, \citet{luo2015alternative} and \citet{bebu2016large} established large-sample inference for the win ratio by leveraging the theory of U-statistics. \citet{mao2019alternative} clarified the null and alternative hypotheses for win ratio as a testing procedure. \citet{buyse2010generalized} developed the generalized pairwise-comparison framework and the net benefit. Later, \citet{dong2023winstatistics} showed that the win ratio, win odds, and net benefit can complement one another as alternative summaries of the same underlying win, loss, and tie probabilities. In comparison, considerably less effort has been made to develop win statistics in CRTs. In a CRT, pairwise comparisons are formed from outcomes that are correlated within clusters, so the effective information is determined by the number of independent clusters rather than the number of individual individuals. As a result, variance estimation must respect that the clusters are the unit of randomization and that individual observations are correlated within clusters \citep{rosner1999use,zou2021confidence}. We are aware of select studies that provided some development of win statistics in CRTs. For example, \citet{zhang2021inference} considered hierarchical time-to-event composites under semi-competing risks and formulated the win ratio as a ratio of clustered U-statistics, and proposed asymptotically valid variance and covariance estimators under within-cluster dependence. \citet{davies2026confidence} developed inference for the win probability in cluster randomized trials with hierarchical composite endpoints by transforming hierarchical pairwise comparisons into individual-level win fractions and then applying a working linear mixed model to obtain cluster-adjusted point and variance estimators. However, this method was restricted to non-censored outcomes and did not address hierarchical time-to-event endpoints. More recently, \citet{fang2025sample} developed a new power and sample size methodology for win statistics under cluster randomization. They described testing procedures for win ratio, win odds, and net benefit based on the analytic randomization-based variances. Inspection of the analytic variance expressions makes explicit the roles of cluster-size variability and a rank-based analog of the intracluster correlation coefficient in determining the number of clusters required to achieve adequate power.

Despite these initial developments, several notable gaps remain in applying win statistics to CRTs. For example, \citet{zhang2021inference} and \citet{fang2025sample} focused solely on one specific win statistic (and summary measure) and one specific approach for conducting statistical inference, but did not cover the whole landscape of win statistics in CRTs. Table \ref{tab:lit_map} summarizes the possible testing procedures for each of the four win summary measures, and only a minor proportion of these possible testing procedures have been previously studied, leaving open the question of the optimal testing procedures that are best applicable to CRTs. Additionally, there has been no attempt to unify win ratio, win odds, net benefit, and DOOR under a common formulation in CRTs, nor has there been a systematic survey of different hypothesis-testing procedures and their finite-sample performance in realistic CRT settings, where the total number of randomized clusters is often limited. Therefore, in this work we address these various gaps by formally developing a range of inferential procedures outlined in Table \ref{tab:lit_map} under parallel-arm CRTs, and compare their empirical performance using simulations.  

\begin{table}[!htbp]
\centering
\caption{A survey and summary of inferential procedures for different win statistics in cluster randomized trials.}
\label{tab:lit_map}
\renewcommand{\arraystretch}{1}
\resizebox{\textwidth}{!}{%
\begin{tabular}{llcccc}
\toprule
Test procedure & Technical details & Win ratio & Win odds & Net benefit & DOOR \\
\midrule
\multirow[t]{2}{*}{\begin{tabular}[t]{@{}l@{}}Wald test\\ \footnotesize Section \ref{sec:wald_test}\end{tabular}}
& Clustered rank-sum Wald test & \citet{fang2025sample}& \citet{fang2025sample}& \citet{fang2025sample} & This paper \\
& Bivariate clustered $U$-statistic Wald test & \citet{zhang2021inference} & This paper & This paper & This paper \\
\midrule
\multirow[t]{2}{*}{\begin{tabular}[t]{@{}l@{}}Randomization-based tests\\ \footnotesize Section \ref{sec:score}\end{tabular}}
& Permutation test & This paper & This paper & This paper & This paper \\
& Cluster FS test & This paper & This paper& This paper & This paper \\
\midrule
\begin{tabular}[t]{@{}l@{}}Jackknife Wald test\\ \footnotesize Section \ref{sec:jk_var}\end{tabular} & Delete-one-cluster jackknife variance & This paper & This paper & This paper & This paper \\
\midrule
\begin{tabular}[t]{@{}l@{}}Likelihood ratio test\\ \footnotesize Section \ref{sec:jel}\end{tabular}
& Jackknife empirical likelihood ratio test & This paper & This paper & This paper & This paper \\
\bottomrule
\end{tabular}
}
\end{table}

The remainder of this paper is organized as follows. Section \ref{sec:win_stat} introduces a common pairwise-comparison framework in which the win ratio, win odds, net benefit, and DOOR are expressed as smooth functions of the same win, loss, and tie probabilities under a prespecified ranking rule. Building on the kernel representation in Section \ref{sec:test}, we develop six inferential procedures for CRTs that share the same underlying estimator but differ in how sampling variability is handled and the final test statistic is constructed, encompassing Wald tests in Section \ref{sec:wald_test} and \ref{sec:jk_var}, randomization-based tests in Section \ref{sec:score}, and jackknife empirical likelihood ratio inference in Section \ref{sec:jel}. Section \ref{sec:sim}  compares their finite-sample performance through extensive simulations, with particular emphasis on type I error rate control and power, to provide practical guidance for their use. Section \ref{sec:real_dat} illustrates these procedures by reanalyzing Strategies to Reduce Injuries and Develop Confidence in Elders (STRIDE), a pragmatic CRT with a composite outcome comprising fatal events and time to first fall-related injury. Section \ref{sec:discussion} concludes with practical recommendations and directions for future research. All methods are implemented in the \texttt{WinsCRT} R package, available on CRAN at \url{https://cran.r-project.org/web/packages/WinsCRT/index.html}

\section{Formulation of Win statistics in cluster-randomized trials}  \label{sec:win_stat}

Let \(M\) denote the number of clusters in a parallel-arm CRT. A proportion \(q \in (0,1)\) of clusters, say \(qM\), are randomized to the intervention and the remaining \((1-q)M\) clusters receive usual care. For cluster \(i \in \{1,\dots,M\}\), let \(N_i\) be the number of individuals in cluster \(i\), and let \(A_i \in \{0,1\}\) denote the cluster-level treatment assignment, with \(A_i = 1\) indicating intervention and \(A_i = 0\) control. Under complete randomization we have \(q = \mathbb{E}(A_i)\). For each individual \(j \in \{1,\dots,N_i\}\) in cluster \(i\), we observe a vector of outcomes \(\bm{Y}_{ij} = (Y_{ij1},\dots,Y_{ijV})^\top\), where \(Y_{ijv} \in \mathcal{Y}_v\) is the \(v\)-th component endpoint and the components are ordered from most to least clinically important. This framework accommodates conventional single-endpoint trials when \(V=1\) as well as hierarchical composite endpoints when \(V \ge 2\). We assume that clusters are independent, so \(\{\bm{Y}_{i1},\ldots,\bm{Y}_{iN_i}\}\perp \{\bm{Y}_{k1},\ldots,\bm{Y}_{kN_k}\}\) for all \(i\neq k\). Within a cluster, however, we allow arbitrary dependence both among the components of a individual's outcome vector and among the outcome vectors of different individuals belonging to that cluster. Let \(n_1 = \sum_{i=1}^M A_i N_i\), \(n_0 = \sum_{i=1}^M (1-A_i) N_i\), and \(n = n_1 + n_0\) denote the total numbers of individuals in the intervention, control, and overall, respectively. To define win summary measures, we consider all pairwise comparisons between individuals from different treatment arms. Specifically,  we take an arbitrary pair consisting of individual \(j\) in cluster \(i\) and individual \(l\) in cluster \(k\), denoted \((i,j)\) and \((k,l)\), with \(A_i \ne A_k\), and then a pre-specified hierarchical clinical ranking rule can be used to compare \(\bm{Y}_{ij}\) and \(\bm{Y}_{kl}\) sequentially along the ordered components until a win (loss) is determined or a tie is reached. That is, we write \(\bm{Y}_{ij} \succ \bm{Y}_{kl}\) if \((i,j)\) has a more favorable outcome than \((k,l)\) according to this rule, \(\bm{Y}_{ij} \prec \bm{Y}_{kl}\) if \((k,l)\) is more favorable, and \(\bm{Y}_{ij} = \bm{Y}_{kl}\) if the pair is tied. For example, in a hierarchical endpoint that ranks all-cause mortality above hospitalization, we say \(\bm{Y}_{ij} \succ \bm{Y}_{kl}\) if individual \((i,j)\) survives but \((k,l)\) dies, or if both survive but \((i,j)\) is hospitalized later (or fewer times) than \((k,l)\). A win is recorded when \(\bm{Y}_{ij} \succ \bm{Y}_{kl}\), a loss is recorded when \(\bm{Y}_{ij} \prec \bm{Y}_{kl}\), and a tie is recorded when neither outcome is preferable \citep{pocock2012win}.

We allow $N_i$ to vary randomly across clusters but assume it is non-informative, meaning that the marginal distribution of individual outcomes under each treatment arm does not depend on cluster size. Consequently, the subsequent individual-pair-level expectations and probabilities are well-defined with respect to the marginal distribution of pairs (one selected from each arm). Consider two distinct clusters \(i\neq k\), independently sampled from this target cluster population, and let \((i,j)\) and \((k,l)\) denote the individuals selected independently and uniformly within their respectively clusters.  We define \(\pi_{\mathrm{win}}=\mathbb{P}(\bm{Y}_{ij}\succ\bm{Y}_{kl}\mid A_i=1,A_k=0)\) as the probability that the intervention individual has a more favorable outcome than the control individual, \(\pi_{\mathrm{loss}}=\mathbb{P}(\bm{Y}_{ij}\prec\bm{Y}_{kl}\mid A_i=1,A_k=0)\) as the probability of the opposite comparison,  and \(\pi_{\mathrm{tie}}=\mathbb{P}(\bm{Y}_{ij}=\bm{Y}_{kl}\mid A_i=1,A_k=0)\) as the probability of a tie under the prespecified comparison rule. Under the non-informative cluster size assumption, the population win estimands defined here admit either a cluster-pair interpretation or the individual-pair interpretation considered in \citet{lee2026s,chen2026optimal}, because equal cluster-pair weighting and individual-pair weighting yield the same marginal win, loss, and tie probabilities. A proof is provided in \ref{app:equiv}. Based on these quantities, we consider four population-level win measures, including the win difference (net benefit), win ratio, win odds, and desirability of outcome ranking (DOOR), which are defined as
\begin{align*}
    W_D &= \pi_{\text{win}} - \pi_{\text{loss}}, \quad  W_R = \frac{\pi_{\text{win}}}{\pi_{\text{loss}}},\\
    W_O &= \frac{\pi_{\text{win}} + 0.5\,\pi_{\text{tie}}}{\pi_{\text{loss}} + 0.5\,\pi_{\text{tie}}}, \quad  \text{DOOR} = \pi_{\text{win}} + 0.5\,\pi_{\text{tie}}.
\end{align*}

Each quantity has a different interpretation. Win difference quantifies the absolute excess probability that a randomly chosen individual from the intervention individual population has a more favorable outcome than a randomly chosen individual from the control individual population. The win ratio is the ratio of the probability of a treatment win to the probability of a treatment loss. Win odds modifies the win ratio by assigning half of the tie probability to wins and half to losses, therefore stabilizing the win statistics when ties are common. DOOR represents the probability that a randomly selected treated individual has a more desirable outcome than a randomly selected control individual, plus half the probability that their outcomes are tied. This quantity coincides with the Mann--Whitney parameter for the underlying hierarchical comparison and can be interpreted as the mean ``desirability score'' for treatment versus control. These four win statistics provide complementary absolute, relative, and probabilistic summaries of the same underlying win/loss/tie structure for hierarchical composite endpoints in CRTs. To construct consistent estimators for the win measures, we aggregate the observed wins, losses, and ties over all cross-arm pairs. That is, we first let
\begin{align*}
W &= \sum_{i=1}^{M} \sum_{j=1}^{N_i} \sum_{k=1}^{M} \sum_{l=1}^{N_k}
A_i (1-A_k)\,\mathbb{I}\{\bm{Y}_{ij} \succ \bm{Y}_{kl}\},\\
L &= \sum_{i=1}^{M} \sum_{j=1}^{N_i} \sum_{k=1}^{M} \sum_{l=1}^{N_k}
A_i (1-A_k)\,\mathbb{I}\{\bm{Y}_{ij} \prec \bm{Y}_{kl}\},\\
T &= \sum_{i=1}^{M} \sum_{j=1}^{N_i} \sum_{k=1}^{M} \sum_{l=1}^{N_k}
A_i (1-A_k)\,\mathbb{I}\{\bm{Y}_{ij} = \bm{Y}_{kl}\},   
\end{align*}
denote the total numbers of wins, losses, and ties, respectively, with \(\mathbb{I}(\cdot)\) the indicator function. By construction, each treated-control pair contributes exactly one of win, loss, or tie, so \(W+L+T = n_1 n_0\). Plug-in estimators of the win-based statistics are then
\[
\widehat{W}_D = \frac{W - L}{n_1 n_0}, \qquad
\widehat{W}_R = \frac{W}{L}, \qquad
\widehat{W}_O = \frac{W + 0.5 T}{L + 0.5 T}, \qquad
\widehat{\text{DOOR}} = \frac{W + 0.5 T}{n_1 n_0}.
\]
The win ratio estimator requires \(L>0\), whereas the win odds estimator requires \(L+0.5T>0\). It is often useful to work on the transformed scales. For the ratio-type win statistics, their relationship with the difference-type win statistics \(W_D\) and the tie probability can be expressed as \(g(W_D, \pi_{\text{tie}})\), for example,
\begin{align} \label{eq:WD_transform}
 \log(\widehat W_R)= 2\,\mathrm{atanh} \left( \frac{\widehat W_D}{1-\widehat{\pi}_{\text{tie}}}\right),\quad \log(\widehat W_O) &= 2\,\mathrm{atanh}(\widehat W_D),\quad \widehat{\text{DOOR}} = \frac{1}{2}\bigl\{1+\widehat W_D\bigr\},   
\end{align}
where \(\mathrm{atanh}(x) = 2^{-1}\log\!\left\{(1+x)/(1-x)\right\}\) for \(-1<x<1\), and \(\widehat{\pi}_{\text{tie}} = T/(n_1 n_0)\) is the empirical tie probability. Thus, all four estimators are constructed from the same empirical win, loss, and tie probabilities. The win odds and DOOR are one-to-one transformations of \(\widehat W_D\), whereas the win ratio additionally depends on the empirical tie probability.

\section{Developing win test statistics}
\label{sec:test}
To construct test statistics for the win difference $W_D=\pi_{\text{win}}-\pi_{\text{loss}}$, we use the hierarchical comparison rule discussed in Section \ref{sec:win_stat}. Let $\mathcal{W}(\bm{Y}_{ij},\bm{Y}_{kl})\in\{-1,0,1\}$ be the signed score produced by this rule for any pair $(i,j)$ and $(k,l)$, where $+1$ indicates a win, $-1$ a loss, and $0$ a tie. Writing $s_{ij,kl}=\mathcal{W}(\bm{Y}_{ij},\bm{Y}_{kl})$ and restricting to treated control pairs leads to the following estimator of WD:
\begin{equation} \label{eq:cross_arm}
    \widehat{W}_D =\frac{1}{n_1n_0}\sum_{i=1}^{M}\sum_{j=1}^{N_i}\sum_{k=1}^{M}\sum_{l=1}^{N_k} A_i(1-A_k)\,s_{ij,kl} =\frac{W-L}{n_1n_0}. 
\end{equation}
By construction, $\mathcal{W}$ is antisymmetric, $s_{ij,kl}=-s_{kl,ij}$ with $s_{ij,ij}=0$, so within-treatment comparisons contribute zero after aggregation. Rearranging the summations gives a rank sum expression in the spirit of the Finkelstein-Schoenfeld approach \citep{finkelstein1999combining},
\begin{equation} \label{eq:cross_all}
\widehat{W}_D =\frac{1}{n_1n_0}\sum_{i=1}^{M}\sum_{j=1}^{N_i}A_i \left\{\sum_{k=1}^{M}\sum_{l=1}^{N_k}s_{ij,kl}\right\}, 
\end{equation}
where the inner sum is the rank accumulated by treated individual $(i,j)$ against all other individuals across clusters. Thus, $\widehat{W}_D$ can be interpreted as the average net score of treated individuals against the full study population, except that the within-treatment contributions are canceled by antisymmetry. Since this net score is computed from pairwise comparisons aggregated over all individuals, \(\widehat{W}_D\) in \eqref{eq:cross_all} is equivalently the average rank of treated individuals in the pooled sample. This representation clarifies the close connection between win difference estimation and classical two-sample rank procedures. In the simplest case $V=1$ with fully observed outcomes, the comparison rule $\mathcal{W}$ induces a complete order over all individuals, where every pair $(i,j)$ and $(k,l)$ can be unambiguously ranked, so $s_{ij,kl}$ reduces to the usual pairwise sign kernel of the Mann-Whitney-Wilcoxon statistic and $\widehat{W}_D$ is a rescaled Wilcoxon rank-sum statistic (or equivalently, $2\,\widehat{\mathrm{DOOR}}-1$). When time-to-event outcomes are subject to censoring, however, $\mathcal{W}(\cdot)$ can no longer resolve every comparison. If both individuals in a pair are censored before a winner can be determined, neither $+1$ nor $-1$ can be assigned and the pair is recorded as a tie ($s_{ij,kl}=0$). The comparison rule thus induces only a partial order over the study population, with some pairs left unranked, and the net score now accumulates over resolvable comparisons only. As a result, $\widehat{W}_D$ in \eqref{eq:cross_all} retains its net score interpretation but now refers to this partial ranking rather than a complete one; see \citet{fang2025sample} for additional discussions. The win statistics framework accommodates this aspect naturally, because $\mathcal{W}$ can incorporate Gehan's rule \citep{gehan1965generalized} for time-to-event components, treating unresolvable pairs as ties, without requiring any modification to the estimator itself.

These properties show that the same estimator $\widehat{W}_D$ admits different, algebraically equivalent decompositions in terms of the kernel $s_{ij,kl}$. In CRTs, however, these decompositions can lead to different inferential procedures, because clusters rather than individual pairs constitute the independent units, and the choice of decomposition determines how one constructs a valid variance estimator and obtains the resulting test statistic for 

\begin{equation}\label{eq:null}
H_0:\tau=\tau_0,\qquad
\tau_0=
\begin{cases}
0, & \tau\in\{W_D,\log(W_R),\log(W_O)\},\\
1/2, & \tau=\mathrm{DOOR}.
\end{cases}
\end{equation}

In what follows, we consider six procedures in three families. The first comprises three Wald tests with different variance estimation strategies. The clustered rank-sum Wald test of Fang, Cao, and Li (FCL) exploits the cluster-score representation of $\widehat{W}_D$ in \eqref{eq:cross_all} to estimate its variance and extends inference to the other win measures through the delta method \citep{fang2025sample}. The bivariate clustered $U$-statistic Wald test of Zhang and Jeong (ZJ) instead treats $(\widehat\pi_{\mathrm{win}},\widehat\pi_{\mathrm{loss}})$ as a bivariate clustered $U$-statistic based on \eqref{eq:cross_arm} and obtains the variance of each win measure from their joint covariance matrix \citep{zhang2021inference}. The delete-one-cluster jackknife Wald test estimates sampling variability by recomputing the selected estimator after removing each cluster in turn, without requiring an analytical variance formula \citep{lee2026s}. The second family consists of two randomization-based tests. The permutation test obtains its reference distribution under the sharp null by reassigning treatment labels across clusters according to the allocation mechanism and recomputing the selected win statistic for each assignment. The cluster Finkelstein--Schoenfeld (FS) test uses the cluster-score representation in \eqref{eq:cross_all}, standardizes the assignment weighted cluster score sum using its analytical randomization variance, and applies a normal approximation to obtain the reference distribution \citep{finkelstein1999combining}. The third family is represented by the jackknife empirical likelihood (JEL) ratio test, which converts the leave-one-cluster recomputations into pseudo-values and uses a Wilks-type Chi-squared reference distribution without specifying a parametric likelihood \citep{peng2018jackknife}. We provide the relevant detail of each approach below.

\subsection{Wald tests} \label{sec:wald_test}

For $\tau\in\{W_D,\log(W_R),\log(W_O),\mathrm{DOOR}\}$ with plug-in estimator $\widehat\tau$ from Section \ref{sec:win_stat}, Wald-type inference is based on \(Z_\tau=\frac{\widehat\tau-\tau_0}{\widehat\sigma_\tau}\), with $\tau_0=0$ under the null hypothesis for $W_D$, $\log(W_R)$, and $\log(W_O)$ and $\tau_0=1/2$ for $\mathrm{DOOR}$. Because each $\widehat\tau$ is a smooth transformation of $(\widehat W_D,\widehat\pi_{\mathrm{tie}})$, we can estimate the variance of $\widehat W_D$ and extend to $\widehat\tau$ by the delta method using the transformations in Section \ref{sec:win_stat}. To obtain cluster-level inference, we use the fact that $\widehat W_D$ admits a cluster-score representation in \eqref{eq:cross_all} \citep{fang2025sample}. Define the centered cluster score \(S_i=\sum_{j=1}^{N_i}\sum_{k=1}^{M}\sum_{l=1}^{N_k}s_{ij,kl}\), so that \(\sum_{i=1}^M S_i=0\). Then \(\widehat W_D=\frac{1}{n_1n_0}\sum_{i=1}^M A_i S_i\), which makes explicit that the $M$ cluster-level scores $S_i$ govern first-order variation under independent clusters. Because $\widehat W_D$ is a pairwise statistic induced by the kernel $s_{ij,kl}$, it admits a Hoeffding--Hájek first-order projection onto cluster-level scores under independent clusters. Thus, the variance can be expressed as the arm-specific second moments of $S_i$. Writing $\sigma_a^2=\mathrm{Var}(S_i\mid A_i=a)$ for $a\in\{0,1\}$, this implies that $\widehat W_D$ is asymptotically normal with variance
\begin{equation}  \label{eq:wald_cluster_score_var}
    \sigma_{W_D}^2=\left(\frac{M q(1-q)}{n_1n_0}\right)^2\left\{\frac{\sigma_1^2}{qM}+\frac{\sigma_0^2}{(1-q)M}\right\}.
\end{equation}
A consistent estimator of $\sigma_{W_D}^2$ is obtained by replacing $\sigma_a^2$ with the sample variance of $S_i$ within arm $A_i=a$. Specifically, let \(\overline{S}_a=\{q^a(1-q)^{1-a}M\}^{-1}\sum_{i=1}^M \mathbb{I}(A_i=a)\,S_i\), \(\widehat\sigma_a^2=\{q^a(1-q)^{1-a}M-1\}^{-1}\sum_{i=1}^M \mathbb{I}(A_i=a)\,(S_i-\overline{S}_a)^2 \), for $a\in\{0,1\}$. Then the variance estimator for $\widehat W_D$ is
\[\widehat\sigma_{W_D}^2=\left(\frac{M q(1-q)}{n_1n_0}\right)^2\left\{\frac{\widehat\sigma_1^2}{qM}+\frac{\widehat\sigma_0^2}{(1-q)M}\right\},
\]
which is valid under arbitrary within-cluster dependence and accommodates unequal cluster sizes through the realized scores $S_i$. Because the remaining win statistics are smooth one-to-one transformations of $(W_D,\pi_{\mathrm{tie}})$ in \eqref{eq:WD_transform}, Wald inference for each $\tau\in\{W_D,\log(W_R),\log(W_O),\mathrm{DOOR}\}$ follows by the delta method. In particular, the corresponding plug-in variance estimators are
\[\widehat\sigma_{\log W_R}^2=\left[\frac{2\{1/(1-\widehat\pi_{\mathrm{tie}})\}}{1-\left(\{1/(1-\widehat\pi_{\mathrm{tie}})\}\widehat W_D\right)^2}\right]^2\widehat\sigma_{W_D}^2,\qquad
\widehat\sigma_{\log W_O}^2=\frac{4\,\widehat\sigma_{W_D}^2}{(1-\widehat W_D^{\,2})^2}.
\]
For $\mathrm{DOOR}$, the linear relation $\widehat{\mathrm{DOOR}}=(1+\widehat W_D)/2$ yields $\widehat\sigma_{\mathrm{DOOR}}^2=\widehat\sigma_{W_D}^2/4$. Thus, the Wald test is constructed by $\widehat\tau/\widehat\sigma_\tau$ and using the standard normal reference distribution (or a $t$ reference with $M-2$ degrees of freedom as a small-sample adjustment). Formal regularity conditions and proofs of these analytical results are provided in  \ref{app:wald}.

An alternative Wald test follows \citet{zhang2021inference}, where the win statistics are treated as bivariate clustered U-statistics in \eqref{eq:cross_arm} and derive the joint asymptotic distribution of their win and loss components (in their original paper, only the win ratio was investigated, and we make extensions to all remaining win measures). Let $M_a=\sum_{i=1}^M \mathbb{I}(A_i=a)$ for $a\in\{0,1\}$ with  $M_0$ and $M_1$ representing the number of control and treated clusters, respectively.  Define the arm-specific mean cluster sizes \(\overline N_a = M_a^{-1}\sum_{i:A_i=a}N_i \) for $a\in\{0,1\}$, so that $\overline N_1=n_1/M_1$ and $\overline N_0=n_0/M_0$. Define two U-statistics for wins and losses, 
\[
U_{\mathrm{win}}=\frac{1}{M_1M_0}\sum_{i:A_i=1}\sum_{k:A_k=0}\sum_{j=1}^{N_i}\sum_{l=1}^{N_k}\mathbb{I}\{\bm Y_{ij}\succ \bm Y_{kl}\},
\qquad
U_{\mathrm{loss}}=\frac{1}{M_1M_0}\sum_{i:A_i=1}\sum_{k:A_k=0}\sum_{j=1}^{N_i}\sum_{l=1}^{N_k}\mathbb{I}\{\bm Y_{ij}\prec \bm Y_{kl}\}.
\]
The corresponding plug-in estimators of the win, loss, and tie probabilities are
\[
\widehat\pi_{\mathrm{win}}=\frac{U_{\mathrm{win}}}{\overline N_1\overline N_0},\qquad
\widehat\pi_{\mathrm{loss}}=\frac{U_{\mathrm{loss}}}{\overline N_1\overline N_0},\qquad
\widehat\pi_{\mathrm{tie}}=1-\widehat\pi_{\mathrm{win}}-\widehat\pi_{\mathrm{loss}}.
\]
Thus, all four estimators in Section \ref{sec:win_stat} can be written as smooth functions of the pair $(\widehat\pi_{\mathrm{win}},\widehat\pi_{\mathrm{loss}})$, and equivalently, of $(\widehat\pi_{\mathrm{win}},\widehat\pi_{\mathrm{loss}},\widehat\pi_{\mathrm{tie}})$:
\[
\widehat{W}_D=\widehat{\pi}_{\mathrm{win}}-\widehat{\pi}_{\mathrm{loss}},\qquad 
\widehat{W}_R=\frac{\widehat{\pi}_{\mathrm{win}}}{\widehat{\pi}_{\mathrm{loss}}},\qquad
\widehat{W}_O=\frac{\widehat{\pi}_{\mathrm{win}}+0.5\,\widehat{\pi}_{\mathrm{tie}}}{\widehat{\pi}_{\mathrm{loss}}+0.5\,\widehat{\pi}_{\mathrm{tie}}},\qquad
\widehat{\mathrm{DOOR}}=\widehat{\pi}_{\mathrm{win}}+0.5\,\widehat{\pi}_{\mathrm{tie}}.
\]
\citet{zhang2021inference} proved that, as $M\to\infty$, the vector
\[
\Bigl(U_{\mathrm{win}}-\overline N_1\overline N_0\pi_{\mathrm{win}},\; U_{\mathrm{loss}}-\overline N_1\overline N_0\pi_{\mathrm{loss}}\Bigr)^\top
\]
is asymptotically mean-zero bivariate normal with covariance matrix admitting the decomposition
\begin{equation}\label{zhang_var}
\Sigma_U= \frac{\overline N_0^{\,2}}{M_1}\Sigma_1+\frac{\overline N_1^{\,2}}{M_0}\Sigma_0,
\end{equation}
where $\Sigma_1$ and $\Sigma_0$ are the treated- and control-arm covariance matrices built from within-cluster sums of centered conditional win/loss projections. Let $\widehat\Sigma_U$ denote a consistent plug-in estimator of $\Sigma_U$ so that
\[
\widehat\Sigma_\pi=\frac{\widehat\Sigma_U}{(\overline N_1\overline N_0)^2} = 
\begin{pmatrix}
\widehat{\mathrm{Var}}(\widehat\pi_{\mathrm{win}}) & \widehat{\mathrm{Cov}}(\widehat\pi_{\mathrm{win}},\widehat\pi_{\mathrm{loss}})\\
\widehat{\mathrm{Cov}}(\widehat\pi_{\mathrm{win}},\widehat\pi_{\mathrm{loss}}) & \widehat{\mathrm{Var}}(\widehat\pi_{\mathrm{loss}})
\end{pmatrix}
\]
is the variance-covariance matrix for $(\widehat\pi_{\mathrm{win}},\widehat\pi_{\mathrm{loss}})^\top$. For any win statistics with a smooth transformation
$\tau=g(\pi_{\mathrm{win}},\pi_{\mathrm{loss}})$ with plug-in estimator $\widehat\tau=g(\widehat\pi_{\mathrm{win}},\widehat\pi_{\mathrm{loss}})$,
the delta method gives \(\widehat\sigma_{\tau}^2=\nabla g(\widehat\pi_{\mathrm{win}},\widehat\pi_{\mathrm{loss}})^\top\,\widehat\Sigma_\pi\,\nabla g(\widehat\pi_{\mathrm{win}},\widehat\pi_{\mathrm{loss}})\), with
\[
\nabla \widehat{W}_D=(1,-1)^\top,\qquad
\nabla \log \widehat{W}_R=\left(\frac{1}{\widehat\pi_{\mathrm{win}}},-\frac{1}{\widehat\pi_{\mathrm{loss}}}\right)^\top,\qquad
\nabla \log \widehat{W}_O=\frac{2}{1-\widehat W_D^{\,2}}(1,-1)^\top,\qquad
\nabla \widehat{\mathrm{DOOR}}=\frac{1}{2}(1,-1)^\top.
\]
To make the matrices $\Sigma_1$ and $\Sigma_0$ explicit, define the win and loss kernels
\[
\phi_{\mathrm{win}}(\bm Y_{ij},\bm Y_{kl})=\mathbb{I}\{\bm Y_{ij}\succ \bm Y_{kl}\},\qquad
\phi_{\mathrm{loss}}(\bm Y_{ij},\bm Y_{kl})=\mathbb{I}\{\bm Y_{ij}\prec \bm Y_{kl}\}.
\]
For each arm, we define the centered first-order Hájek projections of the win and loss kernels by conditioning on the observed outcome of one subject and integrating over a generic draw from the opposite arm.
For a treated subject $(i,j)$ with $A_i=1$, we write
\[
\varphi^{\mathrm{win}}_{1}(\bm Y_{ij})=\mathbb{E}\!\left\{\phi_{\mathrm{win}}(\bm Y_{ij},\bm Y_{kl})\,\middle|\,\bm Y_{ij},\,A_i=1,\,A_k=0\right\}-\pi_{\mathrm{win}},\qquad \varphi^{\mathrm{loss}}_{1}(\bm Y_{ij})= \mathbb{E}\!\left\{\phi_{\mathrm{loss}}(\bm Y_{ij},\bm Y_{kl})\,\middle|\,\bm Y_{ij},\,A_i=1,\,A_k=0\right\}-\pi_{\mathrm{loss}},
\]
where the conditional expectation integrates over a generic control draw $(k,l)$ with $A_k=0$. Similarly, for a control subject $(k,l)$ with $A_k=0$, we write
\[\varphi^{\mathrm{win}}_{0}(\bm Y_{kl})=\mathbb{E}\!\left\{\phi_{\mathrm{win}}(\bm Y_{ij},\bm Y_{kl})\,\middle|\,\bm Y_{kl},\,A_i=1,\,A_k=0\right\}-\pi_{\mathrm{win}},\qquad\varphi^{\mathrm{loss}}_{0}(\bm Y_{kl})=\mathbb{E}\!\left\{\phi_{\mathrm{loss}}(\bm Y_{ij},\bm Y_{kl})\,\middle|\,\bm Y_{kl},\,A_i=1,\,A_k=0\right\}-\pi_{\mathrm{loss}},
\]
where the conditional expectation integrates over a generic treated draw $(i,j)$ with $A_i=1$. Then we define the cluster-level projected sums as
\[
G_i^{\mathrm{win}}=\sum_{j=1}^{N_i}\varphi^{\mathrm{win}}_{A_i}(\bm Y_{ij}),
\qquad
G_i^{\mathrm{loss}}=\sum_{j=1}^{N_i}\varphi^{\mathrm{loss}}_{A_i}(\bm Y_{ij}),
\qquad
\bm G_i=\bigl(G_i^{\mathrm{win}},\,G_i^{\mathrm{loss}}\bigr)^\top,
\]
and the treated- and control-arm covariance matrices in \eqref{zhang_var} can be written as
\[
\Sigma_1=\mathrm{Var}\!\left(\bm G_i\mid A_i=1\right),\qquad
\Sigma_0=\mathrm{Var}\!\left(\bm G_i\mid A_i=0\right).
\]
Under the within-arm exchangeability assumptions of \citet{zhang2021inference}, each entry of $\Sigma_a$ ($a\in\{0,1\}$) decomposes into a within-subject term and a within-cluster (between-subject) term. Let
\(\overline N_a^{(2)}=\frac{1}{M_a}\sum_{i:A_i=a}N_i(N_i-1)\). Then, for $a\in\{0,1\}$,
\[
\mathrm{Var}\!\left(G_i^{\mathrm{win}}\mid A_i=a\right)=\overline N_a\,\mathrm{Var}\!\left\{\varphi^{\mathrm{win}}_{a}(\bm Y_{ij})\mid A_i=a\right\}
+\overline N_a^{(2)}\,\mathrm{Cov}\!\left\{\varphi^{\mathrm{win}}_{a}(\bm Y_{ij}),\varphi^{\mathrm{win}}_{a}(\bm Y_{ij'})\mid A_i=a\right\},
\]
\[
\mathrm{Var}\!\left(G_i^{\mathrm{loss}}\mid A_i=a\right)=\overline N_a\,\mathrm{Var}\!\left\{\varphi^{\mathrm{loss}}_{a}(\bm Y_{ij})\mid A_i=a\right\} +\overline N_a^{(2)}\,\mathrm{Cov}\!\left\{\varphi^{\mathrm{loss}}_{a}(\bm Y_{ij}),\varphi^{\mathrm{loss}}_{a}(\bm Y_{ij'})\mid A_i=a\right\},
\]
\[
\mathrm{Cov}\!\left(G_i^{\mathrm{win}},G_i^{\mathrm{loss}}\mid A_i=a\right)=\overline N_a\,\mathrm{Cov}\!\left\{\varphi^{\mathrm{win}}_{a}(\bm Y_{ij}),\varphi^{\mathrm{loss}}_{a}(\bm Y_{ij})\mid A_i=a\right\}+\overline N_a^{(2)}\,\mathrm{Cov}\!\left\{\varphi^{\mathrm{win}}_{a}(\bm Y_{ij}),\varphi^{\mathrm{loss}}_{a}(\bm Y_{ij'})\mid A_i=a\right\},
\]
where $j\neq j'$ are distinct individuals in the same cluster. Plugging these entries into $\Sigma_1$ and $\Sigma_0$ gives the within-subject and within-cluster contributions for the covariance matrix of $(U_{\mathrm{win}},U_{\mathrm{loss}})$.  \ref{app:wald_u} provides the regularity conditions for the theory of bivariate clustered U-statistic of \citet{zhang2021inference} along with the proof for completeness. For implementation,   \ref{app:wald_u} provides explicit derivations for computing $\widehat\sigma_{W_D}$, $\widehat\sigma_{\log W_R}$, $\widehat\sigma_{\log W_O}$, and $\widehat\sigma_{\mathrm{DOOR}}$ via the delta method.

\subsection{Randomization-based tests} \label{sec:score}

The permutation test and the cluster Finkelstein--Schoenfeld (FS) test are closely connected to the Wald tests in Section \ref{sec:wald_test}, where both are driven by the same cluster-score representation $\widehat{W}_D = \frac{1}{n_1 n_0}\sum_{i=1}^M A_i S_i$, but they differ in how the reference distribution is obtained. Instead of estimating a variance and appealing to a normal or \(t\) approximation, the randomization-based test uses the cluster randomization mechanism directly. The two procedures we consider differ in the version of the null hypothesis \eqref{eq:null} they target. The first is the permutation test, which targets the sharp null of no treatment effect on any individual's outcome distribution, under which the joint distribution of $(\bm{Y}_{i1},\dots,\bm{Y}_{iN_i})$ is invariant to the cluster-level assignment $A_i$. This permutation test is obtained by recomputing the test statistic over all \(\binom{M}{qM}\) treatment-label permutations that preserve exactly \(qM\) treated clusters. The two-sided \(p\)-value is defined as the proportion of permuted test statistics that are at least as extreme as the observed value of \(\widehat{W}_D\), \(\log(\widehat{W}_R)\), \(\log(\widehat{W}_O)\), or \(\widehat{\mathrm{DOOR}}\). 

The second procedure follows the Finkelstein-Schoenfeld approach \citep{finkelstein1999combining} and works directly with the cluster scores \(S_i\). In contrast to the permutation test, this procedure targets the weak null in \eqref{eq:null}, which only requires $\tau = \tau_0$ at the population level and does not restrict higher-order features of the outcome distribution across arms. In this approach, the observed outcomes, and hence the scores \(S_i\), are treated as fixed, and the variability arises only through the treatment assignment mechanism \(A_1,\dots,A_M\). In our design setup, \((A_1,\dots, A_M)\) is a simple random sample without replacement from a finite population containing \(qM\) intervention clusters and \((1-q)M\) control clusters. It follows that the variance of the assignment-weighted contrast \(\sum_{i=1}^M (A_i-q)S_i\) has the closed-form expression
\begin{equation} \label{es:clustered_fs}
   \operatorname{Var}\!\left\{\sum_{i=1}^M (A_i-q)S_i \,\middle|\, S_1,\ldots,S_M\right\} = \frac{q(1-q)M}{M-1}\sum_{i=1}^M S_i^2, 
\end{equation}
using \(\overline S=M^{-1}\sum_{i=1}^M S_i=0\). This leads to the cluster FS test statistic
\[\frac{\sum_{i=1}^M A_i S_i}{\sqrt{\dfrac{q(1-q)M}{M-1}\sum_{i=1}^M S_i^2}}.\]
Unlike the permutation test, this strategy uses the finite-population variance of the assignment-weighted contrast as the basis for inference. Similar test statistics can be constructed for other win statistics (e.g., WR, WO, DOOR) by transforming the permuted statistics using the smooth transformation in \eqref{eq:WD_transform}. Additional derivations and implementation details are provided in  \ref{app:score}.

\subsection{Jackknife Wald test} \label{sec:jk_var}

A fully nonparametric alternative to the analytic variance estimators in Section \ref{sec:wald_test} is obtained via a delete-one-cluster jackknife. Rather than deriving the variance of $\widehat\tau$ through asymptotic projection arguments or the delta method, the jackknife directly estimates sampling variability by measuring how much $\widehat\tau$ changes when each cluster is removed in turn, treating clusters as the independent units throughout. This makes the approach agnostic to the within-cluster dependence structure and automatically accommodates heterogeneous cluster sizes, without requiring any analytic variance formula. For each cluster $i\in\{1,\ldots,M\}$, let $\widehat\tau_{(-i)}$ denote the same plug-in estimator recomputed after removing cluster $i$ and all its individuals, and define the average leave-one-out estimate $\overline\tau_{(-)}=M^{-1}\sum_{i=1}^M\widehat\tau_{(-i)}$. The delete-one-cluster jackknife variance estimator is
\[\widehat\sigma_{\tau,\mathrm{JK}}^{\,2}=\frac{M-1}{M}\sum_{i=1}^M\left(\widehat\tau_{(-i)}-\overline\tau_{(-)}\right)^2.\]
The jackknife studentized Wald statistic is then
\[Z_{\tau,\mathrm{JK}}=\frac{\widehat\tau-\tau_0}{\widehat\sigma_{\tau,\mathrm{JK}}},\]
under the same null as in Section \ref{sec:wald_test}. In practice, $Z_{\tau,\mathrm{JK}}$ may be referenced to a standard normal limit or to a $t$ distribution with $M-2$ degrees of freedom as a small-sample adjustment. A key practical advantage of the jackknife is that it applies uniformly to all four win statistics $\tau\in\{W_D,\log(W_R),\log(W_O),\mathrm{DOOR}\}$ without requiring separate delta-method calculations for each transformation, since the leave-one-out recomputations automatically propagate through any given smooth function defining $\widehat\tau$.

\subsection{Jackknife empirical likelihood ratio test} \label{sec:jel}

Building on the same leave-one-cluster recomputations $\{\widehat\tau_{(-i)}\}_{i=1}^M$ from Section \ref{sec:jk_var}, we also consider a likelihood ratio test through jackknife empirical likelihood (JEL). Rather than studentizing $\widehat\tau$ with a variance estimate as in the Wald approach, JEL constructs a nonparametric likelihood directly from the data and yields a Wilks-type Chi-squared reference distribution without specifying any parametric model \citep{qin1994empirical,jing2009jackknife,peng2018jackknife}. The key step is to convert the leave-one-cluster estimates into pseudo-values that summarize each cluster's first-order contribution to $\widehat\tau$. Specifically, for each cluster $i\in\{1,\ldots,M\}$, define the jackknife pseudo-value
\[\widehat\tau_i = M\widehat\tau-(M-1)\widehat\tau_{(-i)},\]
which inflates the influence of cluster $i$ by contrasting the full-sample estimate with the leave-one-out estimate. By construction, $M^{-1}\sum_{i=1}^M\widehat\tau_i = \widehat\tau$, so these pseudo-values are cluster-level summaries whose average targets $\tau$ in large samples. Treating $\{\widehat\tau_i\}_{i=1}^M$ as $M$ independent cluster-level observations, the JEL for a candidate value $\tau$ maximizes a nonparametric multinomial likelihood subject to a single mean constraint,
\[\mathcal{L}(\tau)=\max_{\{p_i\}}\left\{\prod_{i=1}^M p_i:\;p_i\ge 0,\;\sum_{i=1}^M p_i=1,\;\sum_{i=1}^M p_i(\widehat\tau_i-\tau)=0\right\},\]
with unconstrained maximum $\mathcal{L}(\widehat\tau)$ attained at $p_i=1/M$. The JEL ratio statistic for testing $H_0\colon\tau=\tau_0$ is
\[R(\tau_0)=-2\log\left\{\frac{\mathcal{L}(\tau_0)}{\mathcal{L}(\widehat\tau)}\right\}.\]
The constrained maximizer has the closed form $p_i(\tau_0)=\{M(1+\lambda(\tau_0)(\widehat\tau_i-\tau_0))\}^{-1}$, where the scalar Lagrange multiplier $\lambda(\tau_0)$ solves
\[\sum_{i=1}^M \frac{\widehat\tau_i-\tau_0}{1+\lambda(\tau_0)(\widehat\tau_i-\tau_0)}=0,\]
and gives
\[R(\tau_0)=2\sum_{i=1}^M \log\{1+\lambda(\tau_0)(\widehat\tau_i-\tau_0)\}.\]
Under the conditions that clusters are independent, $M_1/M\to q\in(0,1)$, $\widehat\tau$ admits a nondegenerate first-order Hájek projection at the cluster level with finite second moment, the limiting variance of $\sqrt{M}(\widehat\tau-\tau)$ is positive, and the convex-hull feasibility condition holds so that $\{\widehat\tau_i-\tau_0\}_{i=1}^M$ contains both positive and negative values, a Wilks-type theorem for JEL gives $R(\tau_0)\overset{d}{\longrightarrow}\chi^2_1$ as $M\to\infty$, and the asymptotic level-$\alpha$ test rejects $H_0$ when $R(\tau_0)>\chi^2_{1,1-\alpha}$ \citep{peng2018jackknife}. This likelihood ratio test is free of any variance estimation step and remains valid under arbitrary within-cluster dependence, because the pseudo-values are formed by cluster deletion and the asymptotics are driven entirely by the $M$ independent cluster-level summaries. The cluster-level Wilks theorem, regularity conditions, and implementation details are provided in  \ref{app:jel} for completeness.

\begin{sidewaystable}
\centering
\caption{Summary of testing procedures for win statistics in CRTs.}
\label{tab:test_summary}
\setlength{\tabcolsep}{6pt}
\renewcommand{\arraystretch}{1.2}
\newcolumntype{L}{>{\raggedright\arraybackslash}X}
\begin{tabularx}{\linewidth}{@{}p{3.2cm} p{4.2cm} L p{4.8cm}@{}}
\toprule
Procedure &
Representative quantity &
Test statistic &
Reference distribution \\
\midrule

FCL Wald test & $s_{ij,kl}\in\{-1,0,1\}$ and \(S_i=\sum_{j=1}^{N_i}\sum_{k=1}^{M}\sum_{l=1}^{N_k}s_{ij,kl}\) &
\(Z_\tau=(\widehat\tau-\tau_0)/\widehat\sigma_\tau\), where \(\widehat\sigma_\tau\) is obtained from the cluster-score variance estimator in \eqref{eq:wald_cluster_score_var} &
Standard normal, or \(t\) with \(M-2\) degrees of freedom \\
\addlinespace[2pt]

ZJ Wald test &
\(\phi_{\mathrm{win}}(\bm Y_{ij},\bm Y_{kl})=\mathbb{I}\{\bm Y_{ij}\succ\bm Y_{kl}\}\), \(\phi_{\mathrm{loss}}(\bm Y_{ij},\bm Y_{kl})=\mathbb{I}\{\bm Y_{ij}\prec\bm Y_{kl}\}\), and \((U_{\mathrm{win}},U_{\mathrm{loss}})\) &
\(Z_\tau=(\widehat\tau-\tau_0)/\widehat\sigma_\tau\), where \(\widehat\sigma_\tau^2=\nabla g^\top \widehat\Sigma_\pi \nabla g\) based on the covariance decomposition in \eqref{zhang_var} &
Standard normal, or \(t\) with \(M-2\) degrees of freedom \\
\addlinespace[2pt]

Jackknife Wald test &
Leave-one-cluster estimators \(\widehat\tau_{(-i)}\) &
\(Z_{\tau,\mathrm{JK}}=(\widehat\tau-\tau_0)/\widehat\sigma_{\tau,\mathrm{JK}}\), where \(\widehat\sigma_{\tau,\mathrm{JK}}^{\,2}=\frac{M-1}{M}\sum_{i=1}^M\left(\widehat\tau_{(-i)}-\overline\tau_{(-)}\right)^2\) and \(\overline\tau_{(-)}=\frac{1}{M}\sum_{i=1}^M \widehat\tau_{(-i)}\) &
Standard normal, or \(t\) with \(M-2\) degrees of freedom \\
\addlinespace[2pt]
Permutation test & Observed test statistic recomputed over treatment-label permutations with exactly \(qM\) treated clusters &
For each permuted assignment, recompute \(\widehat W_D\), \(\log(\widehat W_R)\), \(\log(\widehat W_O)\), or \(\widehat{\mathrm{DOOR}}\), and obtain the two-sided permutation \(p\)-value &
Randomization distribution under the sharp null, evaluated by Monte Carlo sampling\\
\addlinespace[2pt]

Cluster FS test & Cluster scores \(S_i\) in Section \ref{sec:score} &
\[
Z_{W_D}=\frac{\sum_{i=1}^M (A_i-q)S_i}{\sqrt{\dfrac{q(1-q)M}{M-1}\sum_{i=1}^M S_i^2}}\]
&
Standard normal, or \(t\) with \(M-2\) degrees of freedom \\
\addlinespace[2pt]

JEL test & Leave-one-cluster pseudo-values \(\widehat\tau_i=M\widehat\tau-(M-1)\widehat\tau_{(-i)}\) &
\(R(\tau_0)=-2\log\{\mathcal{L}(\tau_0)/\mathcal{L}(\widehat\tau)\}\) &
Asymptotic \(\chi^2_1\) \\
\bottomrule
\end{tabularx}
\end{sidewaystable}

\section{Simulation study}\label{sec:sim}

We conducted a simulation study to compare the finite-sample operating characteristics of the testing procedures summarized in Table \ref{tab:test_summary}. For the simulation study, we consider a common two-component semi-competing risks endpoint. This setting provides a useful basis for comparing the procedures but is not intended to exhaust the outcome structures, comparison rules, censoring mechanisms, or cluster size configurations covered by the general framework. Each simulated dataset comprised \(M\) independent clusters. Treatment was assigned at the cluster level by complete randomization with allocation proportion \(q=0.5\) and exactly \(M/2\) treated clusters. For each dataset, we computed \(\widehat\tau\) for \(\tau\in\{W_D,\log(W_R),\log(W_O),\mathrm{DOOR}\}\) and applied each procedure to test \(H_0:\tau=\tau_0\). Across scenarios, we varied the number of clusters \(M\), the coefficient of variation of cluster size \(\mathrm{CV}(N_i)\), the magnitude of cluster-level heterogeneity, the dependence between the two component event times recorded for the same individual, and the target censoring level; the mean cluster size was fixed at \(\overline N=20\). The full list of settings is reported in Table \ref{tab:sim_scenarios}.

For \(V=2\) ordered component outcomes: component \(v=1\) was terminal (e.g., death) and component \(v=2\) was nonterminal (e.g., first hospitalization), which is only observed if it occurs before death and censoring. Cluster-level heterogeneity was induced through event-specific multiplicative frailties \(\gamma_{i1}\) and \(\gamma_{i2}\), independent across clusters,
with \(\gamma_{i1}\sim \mathrm{Gamma}(\alpha_1,\alpha_1)\), \(\gamma_{i2}\sim \mathrm{Gamma}(\alpha_2,\alpha_2)\),
where we use the shape-rate parameterization so that \(\mathbb{E}(\gamma_{iv})=1\) and \(\mathrm{Var}(\gamma_{iv})=1/\alpha_v\), \(v\in\{1,2\}\). Conditional on \((A_i,\gamma_{i1},\gamma_{i2})\), we generated latent event times \((T_{ij1},T_{ij2})\) from marginal proportional hazards models with Weibull baselines. Specifically, for component \(v\in\{1,2\}\) we defined the cause-specific hazard
\begin{equation}\label{eq:hazards_sim}
\lambda_{v}(t\mid A_i,\gamma_{iv})=\gamma_{iv}\,\lambda_{v0}(t)\exp(\theta_v A_i),
\qquad 
\lambda_{v0}(t)=\kappa_v\lambda_v t^{\kappa_v-1},
\end{equation}
so that the corresponding cumulative hazard and survival functions are
\[
\Lambda_{v}(t\mid A_i,\gamma_{iv})=\gamma_{iv}\lambda_v t^{\kappa_v}\exp(\theta_v A_i),
\qquad
S_{v}(t\mid A_i,\gamma_{iv})=\exp\!\left\{-\Lambda_{v}(t\mid A_i,\gamma_{iv})\right\}.
\]

Within-individual dependence between \((T_{ij1},T_{ij2})\) was introduced through a bivariate Gumbel copula linking the marginal survival functions. Let \(\mathcal{G}_{\eta}(\cdot,\cdot)\) denote the Gumbel copula with association parameter \(\eta\ge 1\),
\[
\mathcal{G}_{\eta}(u_1,u_2)
=\exp\!\left[-\left\{(-\log u_1)^{\eta}+(-\log u_2)^{\eta}\right\}^{1/\eta}\right],
\qquad u_1,u_2\in(0,1),
\]
and define the conditional joint survival function of \((T_{ij1},T_{ij2})\) given \((A_i,\gamma_{i1},\gamma_{i2})\) by
\begin{equation}\label{eq:joint_surv_gumbel}
\Pr(T_{ij1}>t_1,\,T_{ij2}>t_2\mid A_i,\gamma_{i1},\gamma_{i2}) = \mathcal{G}_{\eta}\!\left(S_{1}(t_1\mid A_i,\gamma_{i1}),\,S_{2}(t_2\mid A_i,\gamma_{i2})\right).
\end{equation}
The induced conditional joint distribution is
\[
F(t_1,t_2\mid A_i,\gamma_{i1},\gamma_{i2}) = 1-S_{1}(t_1\mid A_i,\gamma_{i1})-S_{2}(t_2\mid A_i,\gamma_{i2}) + \mathcal{G}_{\eta}\!\left(S_{1}(t_1\mid A_i,\gamma_{i1}),\,S_{2}(t_2\mid A_i,\gamma_{i2})\right).
\]
When \(\eta=1\), \(\mathcal{G}_{\eta}(u_1,u_2)=u_1u_2\) and \((T_{ij1},T_{ij2})\) are conditionally independent given
\((A_i,\gamma_{i1},\gamma_{i2})\), while larger \(\eta\) induces stronger positive dependence.

Cluster sizes \(N_i\) were generated from a shifted negative binomial distribution truncated below at 5 and calibrated to satisfy \(\mathbb{E}(N_i)=\overline N\) and the specified \(\mathrm{CV}(N_i)\). Each individual was subject to independent right censoring with \(C_{ij}=\min(C_{ij}^{\ast},\tau_c)\), where \(C_{ij}^{\ast}\sim\mathrm{Exp}(\xi)\), \(\tau_c\) is the administrative end time, and \(\xi\) was calibrated to achieve the target censoring rate and hence different values of \(\pi_{\mathrm{tie}}\). Observed outcomes were then constructed to reflect the semi-competing structure. For the terminal component, \(\widetilde T_{ij1}=\min(T_{ij1},C_{ij})\), and \(\Delta_{ij1}=\mathbb{I}(T_{ij1}\le C_{ij})\). For the nonterminal component, \(\widetilde T_{ij2}=\min(T_{ij2},T_{ij1},C_{ij})\), and \(\Delta_{ij2}=\mathbb{I}\{T_{ij2}<\min(T_{ij1},C_{ij})\}\), so that the nonterminal event is recorded only if it occurs before death and censoring. The observed outcome vector was \(\bm Y_{ij}=\{(\widetilde T_{ij1},\Delta_{ij1}),(\widetilde T_{ij2},\Delta_{ij2})\}\). Pairwise comparisons were induced by a fixed hierarchical rule \(\mathcal{W}\) on generic outcome vectors
\(\bm y=\{(\tilde t_1,\delta_1),(\tilde t_2,\delta_2)\}\) and \(\bm y'=\{(\tilde t_1',\delta_1'),(\tilde t_2',\delta_2')\}\), defined as
\[
\mathcal{W}(\bm y,\bm y')=\mathbb{I}\{(\delta_1,\delta_1')=(1,1),\ \tilde t_1\neq \tilde t_1'\}\,\operatorname{sign}(\tilde t_1-\tilde t_1')+ \mathbb{I}\{(\delta_1,\delta_1')\neq(1,1)\ \text{or}\ \tilde t_1=\tilde t_1'\}\, \mathbb{I}\{(\delta_2,\delta_2')=(1,1),\ \tilde t_2\neq \tilde t_2'\}\,\operatorname{sign}(\tilde t_2-\tilde t_2'),
\]
so that \(\bm y\succ \bm y'\) if \(\mathcal{W}(\bm y,\bm y')=1\), \(\bm y\prec \bm y'\) if \(\mathcal{W}(\bm y,\bm y')=-1\), and \(\bm y=\bm y'\) if \(\mathcal{W}(\bm y,\bm y')=0\). Type I error rate was assessed under the global null \(\theta_1=\theta_2=0\). Power was assessed under a concordant beneficial alternative with \(\theta_1=\log(0.65)\) and \(\theta_2=\log(0.50)\). Baseline Weibull parameters \((\kappa_v,\lambda_v)\) and \(\tau_c\) were fixed across scenarios. 

For each scenario in Table \ref{tab:sim_scenarios}, we generated \(2000\) independent Monte Carlo replicates under both the null and alternative configurations and conducted the testing procedures summarized in Table \ref{tab:test_summary} at the nominal \(\alpha=0.05\) level. Specifically, we conduct Wald test using the analytical variance estimators following the clustered rank sum in \eqref{eq:wald_cluster_score_var} (abbreviated by the first letters of the authors as FCL), bivariate clustered U-statistics in \eqref{zhang_var} (abbreviated as ZJ), and Wald test using the delete-one-cluster jackknife variance estimator in Section \ref{sec:jk_var},  randomization-based test through permutation test, and cluster FS test in Section \ref{sec:score}, and jackknife empirical likelihood ratio test (JEL) in Section \ref{sec:jel}. For the permutation test procedures, the reference distribution was generated by permuting the cluster treatment labels while preserving the randomization scheme with exactly \(M/2\) treated clusters. In each replicate, we used \(2000\) treatment-label permutations and applied a two-sided rejection criterion. We summarized testing performance by the empirical rejection probability across Monte Carlo replicates, which was interpreted as type I error rate under the null configuration and as power under the alternative configuration.

\begin{table}[!htbp]
\centering
\caption{
Simulation scenario design (with fixed \(\overline N=20\)). The 24 scenarios are generated by the factorial product \(M\times \mathrm{CV}(N_i)\times (\alpha_1,\alpha_2)\times \eta\), where \(M\in\{20,100\}\), \(\mathrm{CV}(N_i)\in\{0.3,0.5\}\), \((\alpha_1,\alpha_2)\in\{(2,2),(1,1)\}\) are frailty shape parameters for the terminal and non-terminal components with \(\mathrm{Var}(\gamma_{iv})=1/\alpha_v\), and \(\eta\in\{1,2,4\}\) is the within-individual dependence parameter (Gumbel copula). The table encodes the \(2^3=8\) combinations of the three two-level factors \(\{M,\mathrm{CV}(N_i),(\alpha_1,\alpha_2)\}\) using \(-/+\) coding: \(M(-)=20\), \(M(+)=100\); \(\mathrm{CV}(-)=0.3\), \(\mathrm{CV}(+)=0.5\); \((\alpha_1,\alpha_2)(-)=(2,2)\), \((\alpha_1,\alpha_2)(+)=(1,1)\). Each coded run is crossed with \(\eta\in\{1,2,4\}\), yielding \(8\times 3=24\) scenarios. Every scenario is evaluated under both configurations: null \((\theta_1,\theta_2)=(0,0)\) with \(\pi_{\mathrm{tie}}\in\{35\%,7\%\}\) and alternative \((\theta_1,\theta_2)=\{\log(0.65),\log(0.50)\}\) with \(\pi_{\mathrm{tie}}\in\{41\%,8\%\}\).
}
\label{tab:sim_scenarios}

\setlength{\tabcolsep}{10pt}
\renewcommand{\arraystretch}{1.1}

\begin{tabular}{@{}r c c c c@{}}
\toprule
Run & \(M\) & \(\mathrm{CV}(N_i)\) & \((\alpha_1,\alpha_2)\) & \(\eta\) \\
\midrule
1 & $-$ & $-$ & $-$ & \(1,2,4\) \\
2 & $+$ & $-$ & $-$ & \(1,2,4\) \\
3 & $-$ & $+$ & $-$ & \(1,2,4\) \\
4 & $+$ & $+$ & $-$ & \(1,2,4\) \\
5 & $-$ & $-$ & $+$ & \(1,2,4\) \\
6 & $+$ & $-$ & $+$ & \(1,2,4\) \\
7 & $-$ & $+$ & $+$ & \(1,2,4\) \\
8 & $+$ & $+$ & $+$ & \(1,2,4\) \\
\bottomrule
\end{tabular}
\end{table}

Figure \ref{fig:sim_typei_20} summarizes the empirical type I error rate under the null configuration \((\theta_1,\theta_2)=(0,0)\) across the 24 scenarios listed in Table \ref{tab:sim_scenarios}. Because only \(M=20\) clusters were considered in this setting, the Wald-type test were implemented using a \(t\) reference distribution with \(M-2=18\) degrees of freedom as a small-sample adjustment. Across all four win measures, the two randomization-based tests, namely the permutation test and the cluster FS test, returned the most stable methods and generally remained closest to the nominal level, although no single method was uniformly best in every scenario. While the clustered rank-sum Wald test (FCL) generally remained close to or below the nomial level, with average empirical type I error rate \(0.048\) for both \(W_D\) and DOOR, and more conservative for \(W_R\) and \(W_O\), with corresponding average type I error rate \(0.034\). This finding is consistent with previous comparative work for CRTs with semi-competing risks, where permutation-based inference was found to provide reliable small-sample type I error rate control \citep{li2022comparison}. The bivariate clustered U-statistic Wald test (ZJ) tended to show more noticeable upward size distortion, especially for \(W_D\) and DOOR, whereas the jackknife-studentized Wald test was typically intermediate between FCL and ZJ. The JEL test showed the largest size distortion overall and most often produced the highest empirical type I error rate across scenarios. The ordering of methods was fairly consistent across the four win measures, so that the main finite-sample differences are driven more by the inferential procedure than by the specific choice of win measure. In addition, Figure \ref{fig:sim_typei_20} does not show a strong or fully monotone effect of any single factor across all methods and estimands. Heavier censoring, corresponding to larger \(\pi_{\mathrm{tie}}\), appears to modestly reduce type I error rate for some procedures in some panels, but the pattern is not sufficiently uniform. Similarly, increasing cluster-size variability from \(\mathrm{CV}(N_i)=0.3\) to \(0.5\) does not systematically change the ordering of methods, although the less stable procedures, especially JEL and sometimes ZJ, tend to fluctuate more across these settings. The effect of the within-individual dependence parameter \(\eta\) is also comparatively modest. 

\begin{sidewaysfigure}
    \centering
    \includegraphics[width=0.8\linewidth]{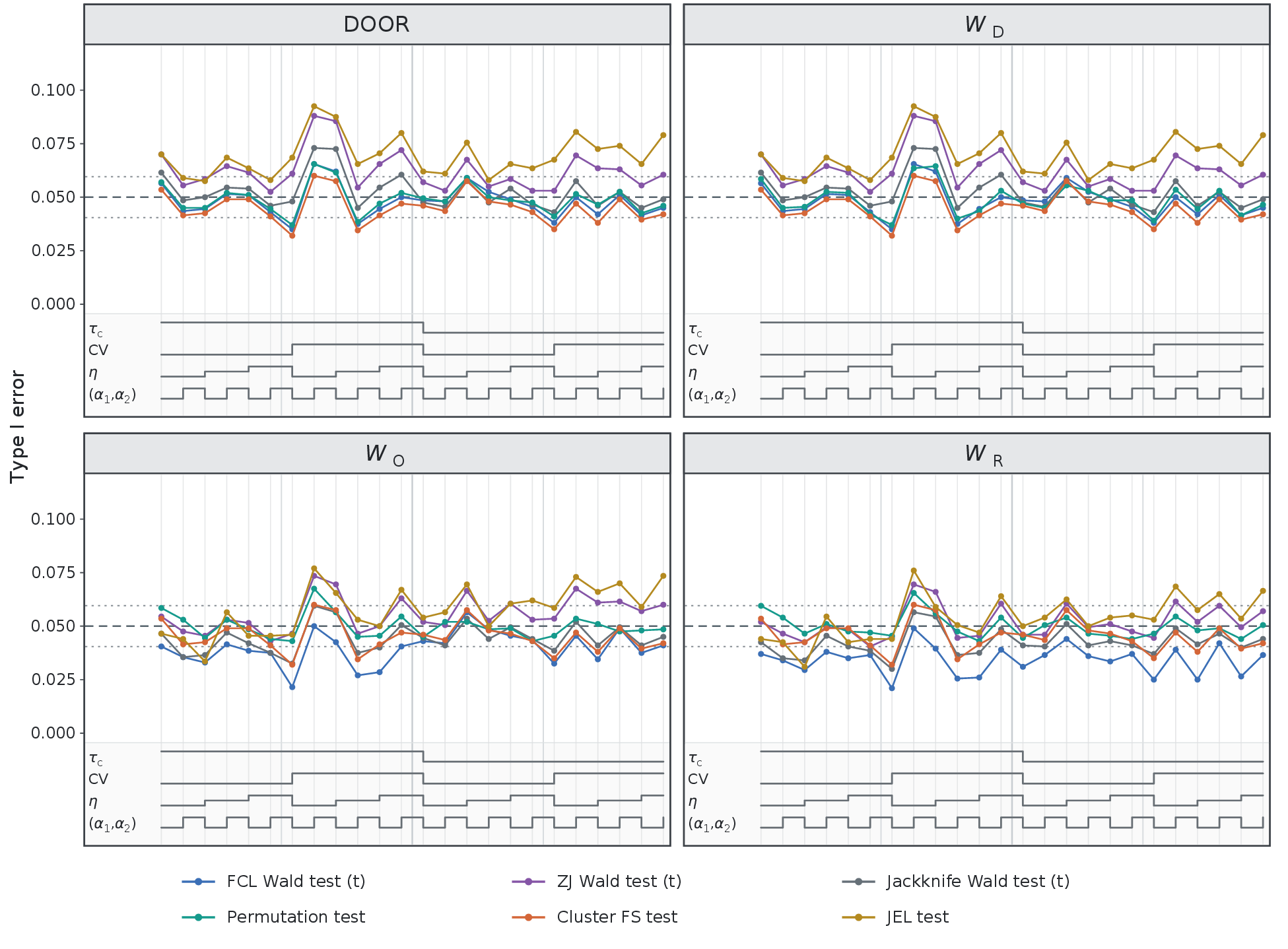}\\
    \caption{Empirical type I error rate for tests of win statistics in parallel-arm cluster-randomized trials under the global null \((\theta_1,\theta_2)=(0,0)\) across Monte Carlo 2,000 replicates. Panel (a) shows \(M=20\) clusters. Within each panel, results are shown for the four win statistics: \(W_D\), \(W_R\), \(W_O\), and \(\mathrm{DOOR}\). The procedure includes Wald t test with df\(=18\) using FCL \citep{fang2025sample}; Wald t test with df\(=18\) using ZJ \citep{zhang2021inference}; Wald t test with df\(=18\) delete-one-cluster jackknife standard errors; permutation test; cluster FS test \citep{finkelstein1999combining}; and the jackknife empirical likelihood (JEL) test. The horizontal dashed line marks the nominal two-sided level \(\alpha=0.05\). The two horizontal dotted lines indicate the Monte Carlo variance band \(0.05 \pm 1.96\sqrt{0.05(1-0.05)/2000}\). The annotation strip above the encodes the scenario factors by: \(\pi_{\text{tie}}\) (35\% versus \(7\%\)), \(\mathrm{CV}(N_i)\) (0.3 versus 0.5 ), within-individual dependence parameter \(\eta\) (1, 2, versus 4), and frailty shape parameters \((\alpha_1,\alpha_2)\) (\((1,1)\) versus (2,2) ).}
    \label{fig:sim_typei_20} 
\end{sidewaysfigure}

Figure \ref{fig:sim_power_20} summarizes the empirical power under the concordant beneficial alternative \((\theta_1,\theta_2)=\{\log(0.65),\log(0.50)\}\) across the 24 scenarios in Table \ref{tab:sim_scenarios}. Overall, power was moderate at this sample size, and the differences among tests were more pronounced than those observed for type I error rate. For the difference-type win measure \(W_D\) and DOOR, the highest average empirical power across scenarios was achieved by JEL and the ZJ Wald test, at \(0.271\) and \(0.269\), respectively, followed by the permutation test and the jackknife Wald test, both at \(0.243\), the FCL Wald test at \(0.239\), and the cluster FS test at \(0.230\). However, this ordering should be interpreted jointly with Figure \ref{fig:sim_typei_20}, because JEL and ZJ also showed the largest type I error rate inflation under the null, whereas FCL and the two randomization-based procedures provided a more favorable balance between power and type I error rate control. For the ratio-type win measure \(\log(W_R)\) and \(\log(W_O)\), the pattern was more distinct. The permutation test was the most powerful procedure across scenarios, with average power \(0.349\) for \(\log(W_R)\) and \(0.339\) for \(\log(W_O)\). The ZJ Wald test was generally the next strongest method, whereas the remaining procedures were less powerful, with FCL tending to be the least powerful for the ratio-type targets, averaging \(0.188\) for \(\log(W_R)\) and \(0.215\) for \(\log(W_O)\). Under the permutation test, ratio type statistics \(\widehat{W}_R\) and \(\widehat{W}_O\) achieve substantially higher empirical power than their difference type statistics \(\widehat{W}_D\) and \(\mathrm{DOOR}\), a pattern that is specific to the permutation test and does not emerge for the remaining procedures, where power differences across statistics are negligible. Thus, in this small-sample setting, the permutation test appears especially attractive for ratio-type win statistics, because it combines comparatively strong power with the most stable type I error rate performance in Figure \ref{fig:sim_typei_20}. In contrast, the lower power of the FCL Wald test for \(W_R\) and \(W_O\) was accompanied by more conservative type I error rates in Figure \ref{fig:sim_typei_20}. For \(W_D\) and DOOR, however, its type I error rate was close to the nominal level and its power was comparable to that of the permutation test. Across scenarios, the variation in power was driven more by the inferential procedure than by any single design factor, although several scenario effects were still visible. For example, under the permutation test, power for \(W_D\) ranged from \(0.121\) to \(0.385\) across the 24 scenarios, and power for \(\log(W_R)\) ranged from \(0.195\) to \(0.515\). The lowest power tended to occur in settings with stronger between-cluster heterogeneity \((\alpha_1,\alpha_2)=(2,2)\) and greater cluster-size variability \(\mathrm{CV}(N_i)=0.5\), whereas the highest power was typically observed when \((\alpha_1,\alpha_2)=(1,1)\) and \(\mathrm{CV}(N_i)=0.3\). Some fluctuations with censoring and the induced tie probability were also apparent, but no single factor produced a completely uniform pattern across all four win measures and all methods.

\begin{sidewaysfigure}
    \centering
    \includegraphics[width=0.8\linewidth]{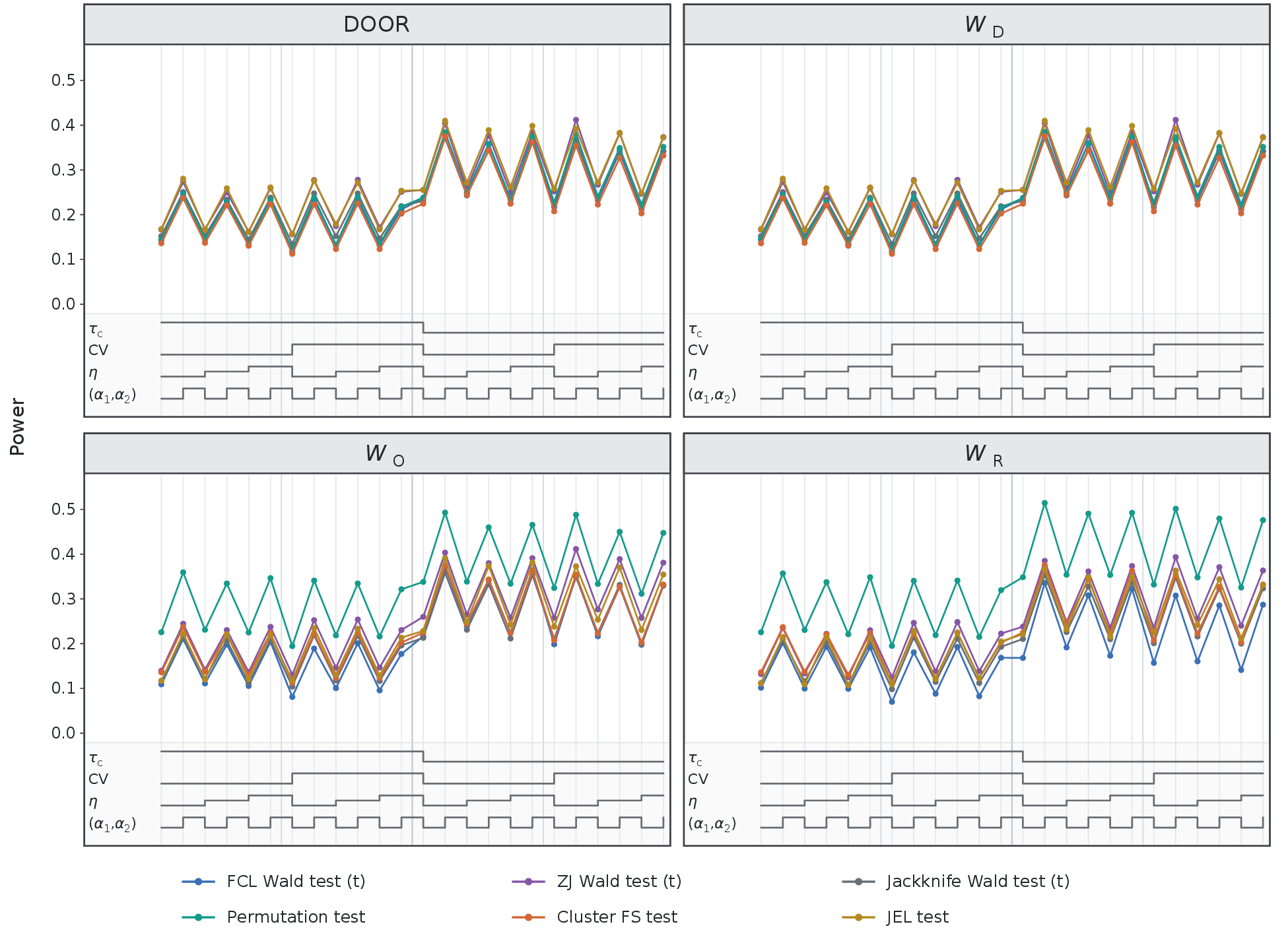}\\
    \caption{Empirical power for tests of win statistics in parallel-arm cluster-randomized trials with \(M=20\) under the concordant beneficial alternative \((\theta_1,\theta_2)=\{\log(0.65),\log(0.50)\}\) across 2{,}000 Monte Carlo replicates. Panel (a) shows \(M=20\) clusters. Within each panel, results are shown for the four win statistics: \(W_D\), \(\log(W_R)\), \(\log(W_O)\), and \(\mathrm{DOOR}\). The procedures include the Wald t test with df\(=18\) proposed using FCL \citep{fang2025sample}; the Wald t test with df\(=18\) using ZJ \citep{zhang2021inference}; a Wald t test with delete-one-cluster jackknife standard errors; the permutation test; the cluster FS test \citep{finkelstein1999combining}; and the jackknife empirical likelihood (JEL) test. Power is computed using a two-sided rejection criterion at the nominal level \(\alpha=0.05\). The annotation strip above the plot encodes the scenario factors by: \(\pi_{\text{tie}}\) (41\% versus 8\%), \(\mathrm{CV}(N_i)\) (0.3 versus 0.5), within-individual dependence parameter \(\eta\) (1, 2, versus 4), and frailty shape parameters \((\alpha_1,\alpha_2)\) (\((1,1)\) versus \((2,2)\)).}
    \label{fig:sim_power_20}
\end{sidewaysfigure}

Web Appendix Figure 2 presents the corresponding empirical type I error rate and power results for a larger number of clusters \(M=100\). For the Wald test, inference in this larger-sample setting was based on the standard normal reference distribution. Compared with the \(M=20\) results, the comparison patterns become much more uniform across procedures, with empirical type I error rate close to the nominal level and substantially improved power for all four win measures (in Web Appendix Figure 4). In particular, the separation among the six methods largely disappears when \(M=100\), and their power curves are nearly indistinguishable across the 24 scenarios. 

Finally, we also examined, for \(M=20\), the effect of using a standard normal reference distribution instead of the \(t\) reference distribution for the Wald-type procedures. The corresponding results are shown in Web Appendix Figure 1 and Web Appendix Figure 3. Compared to Figure 3, the Wald procedures exhibit more noticeable type I error rate inflation under the \(z\) reference distribution, while the randomization-based procedures are unchanged for both the permutation test and the analytical variance. This pattern also agrees with the prior simulation studies \citep{li2022comparison,chen2022finite}, where in CRT with survival outcomes, Wald tests (such those based on marginal Cox proportional hazards model, marginal subdistribution hazards model, and marginal multi-state Cox model) relying on large-sample variance approximations became increasingly anti-conservative as the number of clusters decreased, whereas permutation tests retained satisfactory control of the nominal level. The overall ordering of the Wald procedures is similar, with ZJ remaining the most upwardly shifted, FCL the most stable and conservative among the Wald tests, and the jackknife Wald test lying in between. This additional exploration emphasizes the critical need for small-sample degrees-of-freedom adjustment for studying win measures in CRTs with a small number of clusters. 

\section{An Illustrative Data example} \label{sec:real_dat}

We illustrate the proposed testing procedures using data from the Strategies to Reduce Injuries and Develop Confidence in Elders (STRIDE) trial, a large, pragmatic, parallel-arm, cluster-randomized trial conducted in 86 U.S. primary care practices across 10 healthcare systems. In STRIDE, practices were the unit of randomization and were assigned in a $1{:}1$ allocation to a multifactorial, patient-centered falls-prevention intervention delivered by trained falls-care managers or to enhanced usual care \citep{bhasin2018strategies,bhasin2020randomized}. The trial enrolled $5,451$ community-dwelling adults aged $\ge 70$ years at increased baseline risk of fall injury, including $n_1=2,802$ individuals in intervention practices and $n_0=2,649$ individuals in control practices, with median follow-up $2.35$ months. 
During follow-up, all-cause mortality was $4.26\%$ overall ($4.15\%$ intervention; $4.35\%$ control), while recurrent fall injuries occurred at an overall rate of $44.9$ per 100 person-years ($46.7$ per 100 person-years in intervention; $43.2$ per 100 person-years in control). The primary analysis estimated intervention effects using multi-state survival models with inference that accounted for within-practice correlation. In the original analysis, the intervention reduced the hazard of first individual-reported fall injury (hazard ratio $0.90$, 95\% confidence interval $(0.83,0.99)$, $p=0.004$). We reanalyze STRIDE using win statistics and the hypothesis testing procedures in Section \ref{sec:test}, focusing on a hierarchical composite endpoint with a semi-competing risks structure. Our composite is defined by two time-to-event components ordered by clinical priority, with death as the terminal, higher-priority event and first self-reported fall injury as the nonterminal, lower-priority event. This ordering reflects that death both represents the most severe outcome and truncates the subsequent observation of injuries, so it should dominate treatment comparisons whenever it occurs earlier. Aggregating pairwise comparisons yields counts of wins, losses, and ties, which in turn define the plug-in estimators of net benefit, win ratio, win odds, and DOOR as in Section \ref{sec:win_stat}. We then conduct inference for each target win measure \(\tau\in\{W_D,\log(W_R),\log(W_O),\mathrm{DOOR}\}\) by testing \(H_0:\tau=\tau_0\) at nominal two-sided level \(\alpha=0.05\), implementing the six strategies summarized in Table \ref{tab:test_summary}, where \(\tau_0=0\) for \(W_D, \log(W_R), \log(W_O)\), and \(0.5\) for DOOR.

\begin{table}[!ht]
\centering
\caption{STRIDE illustration results for the four win statistics under the six inferential strategies. Reported quantities include the point estimate (Est.), standard error (SE), and the two-sided \(p\)-value based on the standard normal reference distribution. Standard errors are not defined for the permutation test or the jackknife empirical likelihood (JEL) test.}
\label{tab:stride-test}

\setlength{\tabcolsep}{4pt}
\renewcommand{\arraystretch}{1.12}

\resizebox{\textwidth}{!}{%
\begin{tabular}{
  l
  S[table-format=1.3] S[table-format=1.3] S[table-format=1.3]
  S[table-format=1.3] S[table-format=1.3] S[table-format=1.3]
  S[table-format=1.3] S[table-format=1.3] S[table-format=1.3]
  S[table-format=1.3] S[table-format=1.3] S[table-format=1.3]
}
\toprule
& \multicolumn{3}{c}{$\widehat{W}_D$}
& \multicolumn{3}{c}{$\widehat{W}_R$}
& \multicolumn{3}{c}{$\widehat{W}_O$}
& \multicolumn{3}{c}{$\widehat{\mathrm{DOOR}}$} \\
\cmidrule(lr){2-4}\cmidrule(lr){5-7}\cmidrule(lr){8-10}\cmidrule(lr){11-13}
Method
& {Est.} & {SE} & {$p$-value}
& {Est.} & {SE} & {$p$-value}
& {Est.} & {SE} & {$p$-value}
& {Est.} & {SE} & {$p$-value} \\
\midrule

Wald test (clustered rank sum)
& 0.040 & 0.014 & 0.004
& 1.134 & 0.050 & 0.007
& 1.083 & 0.030 & 0.006
& 0.520 & 0.007 & 0.004 \\

Wald test (bivariate clustered U-statistics)
& 0.040 & 0.013 & 0.002
& 1.134 & 0.047 & 0.004
& 1.083 & 0.028 & 0.003
& 0.520 & 0.006 & 0.002 \\

Wald test (jackknife SE)
& 0.040 & 0.013 & 0.002
& 1.134 & 0.047 & 0.005
& 1.083 & 0.028 & 0.004
& 0.520 & 0.007 & 0.002 \\

Permutation test
& 0.040 & {--} & 0.006
& 1.134 & {--} & 0.004
& 1.083 & {--} & 0.004
& 0.520 & {--} & 0.006 \\

Cluster FS test
& 0.040 & 0.014 & 0.006
& 1.134 & 0.052 & 0.006
& 1.083 & 0.031 & 0.006
& 0.520 & 0.007 & 0.006 \\

Jackknife empirical likelihood (JEL)
& 0.040 & {--} & 0.003
& 1.134 & {--} & 0.005
& 1.083 & {--} & 0.004
& 0.520 & {--} & 0.003 \\

\bottomrule
\end{tabular}%
}

\vspace{2mm}
\footnotesize
The six procedures are the Wald test based on the clustered rank-sum representation (FCL), the Wald test based on bivariate clustered U-statistics (ZJ), the Wald test with delete-one-cluster jackknife standard errors, the permutation test, the cluster FS test, and the jackknife empirical likelihood (JEL) test.
\end{table}

Table \ref{tab:stride-test} summarizes the data analysis results for the four win statistics and the six testing strategies. The overall pattern is consistent across methods. For the win difference, the estimated effect was $\widehat{W}_D=0.040$, with Wald-type standard errors between 0.013 and 0.014 (FCL, ZJ, and jackknife variance) and corresponding two-sided $p$-values between 0.002 and 0.006. The permutation test and cluster permutation test both gave the same conclusion ($p=0.006$), while the jackknife empirical likelihood ratio test also rejected the null ($p=0.003$) at significance level $\alpha=0.05$. A similar pattern can be observed for the ratio-type win statistics. The estimated win ratio was $\widehat{W}_R=1.134$, corresponding to $14.4\%$ more wins than losses in favor of the intervention, and the estimated win odds were $\widehat{W}_O=1.083$. For \(\widehat{W}_R\), the reported \(p\)-value ranged from \(0.004\) to \(0.007\) across the six strategies, which still showed a statistically significant effect. For \(\widehat{W}_O\), the corresponding \(p\)-value ranged from \(0.003\) to \(0.006\). The DOOR statistic was $\widehat{\mathrm{DOOR}}=0.520$, and the evidence was very similar across all six strategies, with \(p\)-values between \(0.002\) and \(0.006\). All six strategies rejected the respective null hypothesis at \(\alpha=0.05\), and the direction of the estimated effects supports an intervention advantage under the hierarchy that prioritizes death over first fall injury, which consider the overall hierarchical contrast instead of a benefit for each component separately. The direction of the benefit is compatible with the original STRIDE finding, although two analysis address different treatment contrasts \citep{bhasin2018strategies}. The original multistate analysis accounted for death and characterize the intervention effect through a transition specific hazard ratio for the first fall injury among participants who remained alive and free of that injury. In contrast, win statistics compare survival first and use first fall injury only when the survival comparison does not determine a winner. Since death were less frequent that first fall injuries, and mortality was similar between the treatment group ($4.26\%$ for both arms), this provides a plausible explanation for the concordant direction of the findings. Although death retains the highest clinical priority, first fall injury may provide substantial discrimination among pairs not resolved by survival. Thus, and the component assigned the highest priority need not contribution most to the overall excess of wins over loss.  
The data analysis results based on the Wald test with a \(t\)-reference distribution and \(84\) degrees of freedom are reported in Web Appendix Table 1. Because the number of clusters is relatively large, these results are very similar to those based on the standard normal reference distribution, and the resulting conclusions are unchanged.


\section{Discussion} \label{sec:discussion}

In this paper, we compared six inferential procedures based on win statistics in parallel-arm CRTs with hierarchical composite endpoints. The procedures share the same prespecified pairwise comparison rule and the same empirical win, loss, and tie proportions, but differ in how independent cluster-level information is represented and calibrated for inference. Specifically, we considered the clustered rank-sum Wald test of Fang, Cao, and Li (FCL), the bivariate clustered $U$-statistic Wald test of Zhang and Jeong (ZJ), the delete-one-cluster jackknife Wald test, the permutation test, the clustered Finkelstein--Schoenfeld (FS) test, and the jackknife empirical likelihood (JEL) ratio test. This common formulation permits direct comparison of procedures that had previously been developed separately or only for selected win measures. The main findings and practical recommendations are summarized in Table \ref{tab:recommendations}.

The four win measures describe treatment benefit on different scales. The win difference summarizes the absolute excess of win over loss probability, and DOOR expresses the win probability with half credit for ties. The win ratio and win odds provide relative summaries, with the latter also assigning half of the tie probability to each group. As discussed by \citet{dong2023winstatistics} these summaries provide complementary interpretations of the same underlying pairwise comparisons. For the win difference and DOOR, the clustered rank-sum Wald test and the two randomization-based tests provided similar power and generally maintained type I error rate close to the nominal level at $M=20$. Among the Wald procedures, the FCL Wald test with a $t$-distribution and $M-2$ degrees of freedom had the most reliable type I error rate control. The bivariate clustered $U$-statistic Wald test and the jackknife empirical likelihood ratio test had higher rejection probabilities under the alternative, but also showed greater type I error rate inflation. The jackknife Wald test was generally intermediate in its type I error rate performance. These findings support the FCL Wald test as an analytic procedure for the win difference and DOOR in small CRTs, with the permutation test providing alternatives with comparable performance in the settings studied. For the win ratio and win odds, the cluster rank-sum Wald test was more conservative, with empirical type I error rate generally below the nominal level. This conservatism was accompanied by the lower power than the permutation test, which maintain type I error rate close to nominal and achieve the higher power in the settings studied. We therefore recommend the permutation for these ratio measure in small CRTs. For all four measures, differences in type I error rate and power were substantially smaller at $M=100$. With a larger number of clusters, we recommend the analytic Wald inference procedure, since it provides convenient yet reliable standard errors and confidence intervals without repeated computationally intensive permutations.

These findings also have implications for study design and sample size determination. \citet{fang2025sample} recently developed analytic power and sample size calculations for the clustered rank-sum Wald test that incorporate within-cluster dependence, cluster-size variability, and the comparison probabilities induced by the endpoint hierarchy and censoring. In their simulations, it was found that the predicted power can sometimes be lower than empirical power in the settings studied, and hence they offered a conservative approach to study design. This conservativeness concerns the variance approximation used for planning and is distinct from the type I error of the test. For the win difference and DOOR, the comparable power of the FCL Wald test and the randomization-based tests supports the use of FCL-based sample size calculations. For the win ratio and win odds, The lower power of the corresponding FCL Wald tests suggests that FCL-based calculations can provide a conservative analytic approach for sample size determination in CRTs with these ratio measures.

\begin{table}[ht]
\caption{Summary of finite-sample behavior and practical recommendations for testing procedures for win statistics in parallel-arm CRTs. FCL = Wald test with clustered rank sum \citep{fang2025sample}; ZJ = Wald test with bivariate clustered $U$-statistic \citep{zhang2021inference}; JEL = jackknife empirical likelihood ratio \citep{jing2009jackknife,peng2018jackknife}. For Wald procedures, a $t$ reference distribution with $M-2$ degrees of freedom is recommended when $M$ is small. For ratio-type win measures, inference should be conducted on $\log(\widehat W_R)$ and $\log(\widehat W_O)$ and then back-transformed for reporting.}
\label{tab:recommendations}
\centering
\resizebox{\textwidth}{!}{%
\begin{tabular}{p{3cm} p{3.3cm} p{3.3cm} p{3.3cm} p{3.5cm}}
\toprule
\textbf{Procedure} & \multicolumn{2}{c}{\textbf{Small $M$}} & \textbf{Large $M$} & \textbf{Recommendation} \\
\cmidrule(lr){2-3} \cmidrule(lr){4-4}
& \textbf{Type I error rate} & \textbf{Power} & \textbf{Type I error rate} &  \\
FCL Wald test \citep{fang2025sample} & Near nominal for $W_D$ and DOOR. Slightly conservative for $W_R$ and $W_O$ & Comparable to the permutation test for $W_D$ and DOOR, but lower for $W_R$ and $W_O$ & Near nominal & Preferred for $W_D$ and DOOR at small $M$. Convenient analytic option for all measures at large $M$ \\
\addlinespace
ZJ Wald test \citep{zhang2021inference} & Can exceed nominal, particularly for $W_D$ and DOOR & Higher rejection rate under the alternative should be interpreted alongside type I error inflation & Occasionally anti-conservative & Not preferred at small $M$, but may be used at large $M$ \\
\addlinespace
Jackknife Wald test & Generally intermediate between FCL and ZJ & Comparable to FCL for $W_D$ and DOOR & Near nominal & Useful when an analytic variance formula is impractical. Use cautiously at small $M$ \\
\addlinespace
Permutation test & Near nominal & Comparable to FCL for $W_D$ and DOOR. Higher than FCL for $W_R$ and $W_O$ & Near nominal &Preferred for $W_R$ and $W_O$ at small $M$, applicable at large $M$ if computational feasible \\
\addlinespace
Cluster FS test \citep{finkelstein1999combining} & Near nominal, slighly nominal & Comparable to FCL for $W_D$ and DOOR. Lower than the permutation test for $W_R$ and $W_O$ & Near nominal & Computationally efficient randomization-based alternative, particularly for $W_D$ and DOOR.  \\
\addlinespace
JEL test \citep{jing2009jackknife,peng2018jackknife} & Inflation is most pronounced for $W_D$ and DOOR & Higher rejection under the alternative & Near nominal & Use cautiously at small $M$, but has reasonable performance at large $M$ \\
\bottomrule
\end{tabular}%
}
\end{table}

Several limitations and directions for future research merit consideration. First, we focus on conventional win statistics, which assign equal weight to small and large resolved wins. This is a limitation when the scientific objective requires larger benefits to contribute more to the treatment comparison. Alternative comparison rules can incorporate prespecified clinical thresholds \citep{buyse2010generalized} or magnitude-based weights. For example, \citet{mou2026generalizing} recently extended the Finkelstein-Schoenfeld test for individually randomized trials by alternating comparisons across endpoints and revisiting unresolved pairs using progressively smaller thresholds. This approach allows clinically meaningful differences in lower-priority outcomes, such as time to hospitalization, to determine a win before small differences in survival are considered. Extending such methods to CRTs warrants further investigation. Second, pragmatic CRTs can include heterogeneous clusters, and in challenging cases, extreme cluster size imbalance can allow a few large clusters to dominate the all-pairs estimator without increasing the number of independent randomized units, potentially compromising the approximations underlying the Wald, cluster FS, and JEL tests. A permutation test calibrated to the actual assignment mechanism avoids these approximations for sharp-null testing, although extreme cluster size imbalance can be a source of power loss. In addition, our framework currently assumes non-informative cluster size; when outcome distributions depend on cluster size, individual-pair and cluster-pair estimands need not coincide, and the weighting scheme should reflect the intended population comparison.\cite{lee2026s,chen2026optimal} Examining the cluster size distribution and changes in estimates after omitting individual clusters can help assess sensitivity to influential clusters. Third, the present work also focuses on parallel-arm CRTs. Generalization to stepped-wedge CRTs would additionally require accounting for time effects and can build on existing work on win odds \citep{bard2026analysis}. Finally, our simulations considered a two-component semi-competing risks endpoint under selected cluster size and censoring configurations; more complex hierarchies, recurrent-event components, heavier censoring, and more extreme cluster size imbalance warrant further investigation.

To assist implementation, the \texttt{WinsCRT} R package is available on CRAN at \url{https://cran.r-project.org/web/packages/WinsCRT/index.html}; a tutorial is provided in Web Appendix 8.

\section*{Acknowledgements}
F.L., G.T., F.P.W. and P.J.H. are supported by the United States National Institutes of Health (NIH), National Heart, Lung, and Blood Institute (NHLBI, grant number 1R01HL178513). All statements in this report, including its findings and conclusions, are solely those of the authors and do not necessarily represent the views of the NIH. The STRIDE study was funded primarily by the Patient Centered Outcomes Research Institute (PCORI\textsuperscript{\textregistered}), with additional support from the National Institute on Aging (NIA) at NIH. Funding is provided and the award managed through a cooperative agreement (5U01AG048270) between the NIA and the Brigham and Women’s Hospital. The authors thank Professor Peter Peduzzi for help in accessing the STRIDE data.


\section*{Supplementary Material}

The supplementary material includes technical derivations and a tutorial for the R package \texttt{WinsCRT} referenced in the article.

\section*{Data Availability Statement}
An R package implementing our method is available at \url{https://cran.r-project.org/web/packages/WinsCRT/index.html}. The STRIDE data can be obtained via the National Institute on Aging (NIA) Aging Research Biobank at \url{https://agingresearchbiobank.nia.nih.gov/studies/stride/details}.

\clearpage
\newpage
\printbibliography

\newpage
\appendix

\setcounter{section}{0}
\renewcommand{\thesection}{A.\arabic{section}}
\renewcommand{\thesubsection}{\thesection.\arabic{subsection}}
\renewcommand{\thesubsubsection}{\thesubsection.\arabic{subsubsection}}

\section{Equivalence of participant-pair and cluster-pair weighting under non-informative cluster size} \label{app:equiv}

We establish that, under non-informative cluster size and the sampling and moment conditions stated below, the marginal win probability \(\pi_{\mathrm{win}}\) defined in Section \ref{sec:win_stat} equals both the cluster-pair-weighted and participant-pair-weighted win probabilities in \citet{lee2026s,chen2026optimal}. For two distinct clusters \(i\neq k\), define the total number of winning participant pairs as
\begin{align*}
H_{ik}=\sum_{u=1}^{N_i}\sum_{v=1}^{N_k} \mathbb{I}\{\bm{Y}_{iu}\succ\bm{Y}_{kv}\}.
\end{align*}
The cluster-pair-weighted and participant-pair-weighted probabilities
are, respectively,
\begin{align*}
\pi_{\mathrm{win}}^{\mathrm{C}}
&=\mathbb{E}\!\left[\frac{H_{ik}}{N_iN_k}\,\middle|\, A_i=1,A_k=0\right],\\
\pi_{\mathrm{win}}^{\mathrm{P}}
&=\frac{\mathbb{E}(H_{ik}\mid A_i=1,A_k=0)}{\mathbb{E}(N_iN_k\mid A_i=1,A_k=0)}.
\end{align*}
The first definition averages the within-cluster-pair win fraction under the unweighted cluster sampling law. The second weights cluster pairs in proportion to their numbers of participant pairs. Thus we can prove that, under non-informative cluster size,
\begin{align*}
\pi_{\mathrm{win}} &=\pi_{\mathrm{win}}^{\mathrm{C}} =\pi_{\mathrm{win}}^{\mathrm{P}}.
\end{align*}
The corresponding equalities for losses and ties then imply equality of the population win measures constructed from these probabilities.

We assume that \(\mathbb{P}(A_i=1,A_k=0)>0\), that cluster sizes are positive finite integers almost surely, and that \(0<\mathbb{E}(N_iN_k\mid A_i=1,A_k=0)<\infty\). Conditional on the indicated treatment assignments, the two cluster records are independent, and each follows the superpopulation distribution corresponding to its own treatment arm. Within-cluster
dependence is unrestricted. As in Section \ref{sec:win_stat}, \((i,j)\) and \((k,l)\) denote participants selected independently and uniformly conditional on the complete cluster records and treatment assignments. The indices \(u\) and \(v\) instead index all participants in the sums above. The comparison rule is fixed and measurable, with win, loss, and tie forming mutually exclusive and exhaustive possibilities. All probabilities and expectations concern this superpopulation sampling law only. Let \(F_a\) denote the marginal outcome distribution of a uniformly selected participant from a cluster in arm \(a\in\{0,1\}\). Non-informative cluster size means that, for every measurable outcome set \(B\) and every cluster size \(m\) with positive probability in that arm,
\begin{align*}
\mathbb{P}(\bm{Y}_{ij}\in B\mid A_i=a,N_i=m) =\mathbb{P}(\bm{Y}_{ij}\in B\mid A_i=a) =F_a(B).
\end{align*}
This condition concerns the full marginal outcome distribution entering the comparison rule. It does not require independence among outcomes within the same cluster. Since \(0\leq H_{ik}\leq N_iN_k\), the moment condition above ensures that all expectations defining the two win probabilities are finite.

\begin{proof}

Write \(C_i=(N_i,\bm{Y}_{i1},\ldots,\bm{Y}_{iN_i})\) for the record of cluster \(i\), and define \(C_k\) analogously. Conditional on the two records and assignments, each participant pair has selection
probability \(1/(N_iN_k)\). Therefore,
\begin{align*}
&\mathbb{P}\!\left(\bm{Y}_{ij}\succ\bm{Y}_{kl}\,\middle|\, C_i,C_k,A_i=1,A_k=0\right)\\
&\quad=\sum_{u=1}^{N_i}\sum_{v=1}^{N_k}\frac{1}{N_iN_k}\,\mathbb{I}\{\bm{Y}_{iu}\succ\bm{Y}_{kv}\}\\
&\quad=\frac{H_{ik}}{N_iN_k}.
\end{align*}
By the law of iterated expectation,
\begin{align*}
\pi_{\mathrm{win}}
&=\mathbb{E}\!\left[\mathbb{P}\!\left(\bm{Y}_{ij}\succ\bm{Y}_{kl}\,\middle|\, C_i,C_k,A_i=1,A_k=0\right)\,\middle|\, A_i=1,A_k=0\right]\\
&=\mathbb{E}\!\left[\frac{H_{ik}}{N_iN_k}\,\middle|\, A_i=1,A_k=0\right]\\
&=\pi_{\mathrm{win}}^{\mathrm{C}}.
\end{align*}
This equality follows from uniform participant selection and does not itself require non-informative cluster size.

For the participant-pair interpretation, sampling a participant from an arm-specific participant population weights each cluster record in proportion to its size. Independent sampling from the two
participant populations therefore weights a cluster pair in proportion
to \(N_iN_k\). Under conditional between-cluster independence,
\begin{align*}
\mathbb{E}(N_iN_k\mid A_i=1,A_k=0)=\mathbb{E}(N_i\mid A_i=1)\, \mathbb{E}(N_k\mid A_k=0).
\end{align*}
Thus, the normalized cluster-pair weight is \(N_iN_k/\mathbb{E}(N_iN_k\mid A_i=1,A_k=0)\). Averaging the conditional win fraction under this reweighted sampling law gives
\begin{align*}
\mathbb{E}\!\left[\frac{N_iN_k}{\mathbb{E}(N_iN_k\mid A_i=1,A_k=0)}\frac{H_{ik}}{N_iN_k}\,\middle|\, A_i=1,A_k=0\right]
&=\frac{\mathbb{E}(H_{ik}\mid A_i=1,A_k=0)}{\mathbb{E}(N_iN_k\mid A_i=1,A_k=0)}\\
&=\pi_{\mathrm{win}}^{\mathrm{P}}.
\end{align*}
The two formulations correspond to different weighting schemes. It remains to show that they yield the same probability under non-informative cluster size.

For any measurable sets \(B_1,B_0\) and any size pair \((m,r)\) with positive conditional probability, between-cluster independence and non-informative cluster size imply
\begin{align*}
&\mathbb{P}\!\left(\bm{Y}_{ij}\in B_1,\bm{Y}_{kl}\in B_0\,\middle|\, N_i=m,N_k=r,A_i=1,A_k=0\right)\\
&\quad=\mathbb{P}(\bm{Y}_{ij}\in B_1\mid N_i=m,A_i=1)\,\mathbb{P}(\bm{Y}_{kl}\in B_0\mid N_k=r,A_k=0)\\
&\quad=F_1(B_1)F_0(B_0).
\end{align*}
The conditional joint distribution of the selected outcomes is therefore \(F_1\otimes F_0\), so
\begin{align*}
&\mathbb{P}\!\left(\bm{Y}_{ij}\succ\bm{Y}_{kl}\,\middle|\, N_i=m,N_k=r,A_i=1,A_k=0\right)\\
&\quad=\int\!\!\int\mathbb{I}\{\bm{y}\succ\bm{z}\}\,F_1(d\bm{y})F_0(d\bm{z}).
\end{align*}
Because this integral does not depend on \(m\) or \(r\), averaging
over the cluster sizes yields
\begin{align*}
\pi_{\mathrm{win}}
&=\mathbb{E}\!\left[\mathbb{P}\!\left(\bm{Y}_{ij}\succ\bm{Y}_{kl}\,\middle|\, N_i,N_k,A_i=1,A_k=0\right)\,\middle|\, A_i=1,A_k=0\right]\\
&=\int\!\!\int\mathbb{I}\{\bm{y}\succ\bm{z}\}\,F_1(d\bm{y})F_0(d\bm{z}).
\end{align*}
In particular, the conditional win probability given the cluster sizes equals \(\pi_{\mathrm{win}}\) almost surely. Combining this identity with the conditional probability given the complete records and applying iterated expectation gives
\begin{align*}
&\mathbb{E}\!\left[\frac{H_{ik}}{N_iN_k}\,\middle|\, N_i,N_k,A_i=1,A_k=0\right]\\
&\quad=\mathbb{E}\!\left[\mathbb{P}\!\left(\bm{Y}_{ij}\succ\bm{Y}_{kl}\,\middle|\, C_i,C_k,A_i=1,A_k=0\right)\,\middle|\, N_i,N_k,A_i=1,A_k=0\right]\\
&\quad=\mathbb{P}\!\left(\bm{Y}_{ij}\succ\bm{Y}_{kl}\,\middle|\, N_i,N_k,A_i=1,A_k=0\right)\\
&\quad=\pi_{\mathrm{win}}.
\end{align*}
Since \(N_iN_k\) is fixed under conditioning on the two sizes,
\begin{align*}
\mathbb{E}(H_{ik}\mid N_i,N_k,A_i=1,A_k=0)
&=N_iN_k\,\mathbb{E}\!\left[\frac{H_{ik}}{N_iN_k}\,\middle|\, N_i,N_k,A_i=1,A_k=0\right]\\
&=N_iN_k\,\pi_{\mathrm{win}}.
\end{align*}
Taking expectation over the cluster sizes gives
\begin{align*}
\mathbb{E}(H_{ik}\mid A_i=1,A_k=0)
&=\mathbb{E}\!\left[\mathbb{E}(H_{ik}\mid N_i,N_k,A_i=1,A_k=0)\,\middle|\, A_i=1,A_k=0\right]\\
&=\mathbb{E}(N_iN_k\,\pi_{\mathrm{win}}\mid A_i=1,A_k=0)\\
&=\pi_{\mathrm{win}}\,\mathbb{E}(N_iN_k\mid A_i=1,A_k=0).
\end{align*}
Dividing by the finite positive expected number of participant pairs
therefore establishes
\begin{align*}
\pi_{\mathrm{win}}^{\mathrm{P}}
&=\frac{\mathbb{E}(H_{ik}\mid A_i=1,A_k=0)}{\mathbb{E}(N_iN_k\mid A_i=1,A_k=0)}\\
&=\frac{\pi_{\mathrm{win}}\,\mathbb{E}(N_iN_k\mid A_i=1,A_k=0)}{\mathbb{E}(N_iN_k\mid A_i=1,A_k=0)}\\
&=\pi_{\mathrm{win}}=\pi_{\mathrm{win}}^{\mathrm{C}}.
\end{align*}
The factorization relies on the constant conditional mean of \(H_{ik}/(N_iN_k)\) given the cluster sizes. It does not require independence of this realized win fraction and the sizes.

For losses and ties, define
\begin{align*}
H_{ik}^{\mathrm{loss}}
&=\sum_{u=1}^{N_i}\sum_{v=1}^{N_k}\mathbb{I}\{\bm{Y}_{iu}\prec\bm{Y}_{kv}\},\\
H_{ik}^{\mathrm{tie}}
&=\sum_{u=1}^{N_i}\sum_{v=1}^{N_k}\mathbb{I}\{\bm{Y}_{iu}=\bm{Y}_{kv}\},
\end{align*}
where equality denotes a tie under the comparison rule. The preceding argument applies to either indicator without modification. Uniform selection gives the corresponding conditional fraction, whose conditional mean is constant across cluster sizes under non-informative cluster size. Thus, for \(s\in\{\mathrm{loss},\mathrm{tie}\}\),
\begin{align*}
\pi_s^{\mathrm{C}}
&=\mathbb{E}\!\left[\frac{H_{ik}^{s}}{N_iN_k}\,\middle|\, A_i=1,A_k=0\right]=\pi_s,\\
\pi_s^{\mathrm{P}}
&=\frac{\mathbb{E}(H_{ik}^{s}\mid A_i=1,A_k=0)}{\mathbb{E}(N_iN_k\mid A_i=1,A_k=0)}=\pi_s.
\end{align*}
As a result, the win difference, win ratio, win odds, and DOOR constructed from these probabilities agree under the two weighting schemes wherever they are defined; the win ratio requires \(\pi_{\mathrm{loss}}>0\), and the win odds requires \(\pi_{\mathrm{loss}}+\tfrac12\pi_{\mathrm{tie}}>0\). This proves the claimed population-level equivalence.

\end{proof}

\section{Regularity conditions and proofs for the clustered rank-sum Wald test (FCL)}
\label{app:wald}

The rank-sum and cross-arm representations in Section \ref{sec:wald_test} satisfy
\[
\widehat W_D=\frac{1}{n_1n_0}\sum_{i:A_i=1}S_i,
\qquad \sum_{i=1}^M S_i=0,
\]
and $S_i$ contains comparisons with other clusters. We first derive the sampling limit from projections onto independent cluster records and then give sufficient conditions for the pooled score variance estimator used by FCL.

Write $\mathcal C_i=(N_i,\bm Y_{i1},\ldots,\bm Y_{iN_i})$, $M_a=\sum_i\mathbb I(A_i=a)$, $q_M=M_1/M$, $\overline N_a=n_a/M_a$, and $\mu_a=\mathbb E(N_i\mid A_i=a)$. All sampling limits in this section, \ref{app:wald_u}, and \ref{app:jel} condition on the allocation vector and concern superpopulation sampling of cluster records. The finite-population randomization argument is given separately in \ref{app:score}.

\begin{enumerate}
\item[(C1)] Conditional on the allocation vector, cluster records are mutually independent and identically distributed within each treatment arm. Their arm-specific laws are fixed as $M$ increases. Dependence within each cluster is unrestricted.
\item[(C2)] $q_M\to q\in(0,1)$, so $M_a\to\infty$ for both arms.
\item[(C3)] Cluster sizes are positive integers, $\mathbb E(N_i^{2+\delta}\mid A_i=a)<\infty$ for some $\delta>0$, and cluster size is non-informative as introduced in \ref{app:equiv}. The signed comparison kernel is fixed, measurable, antisymmetric, and bounded by one in absolute value. The limiting variance $V_{W_D}$ defined below is strictly positive.
\item[(C4)] For inference on the ratio scales, the population probabilities are interior, with \(1-\pi_{\mathrm{tie}}>0\), \(\left|\frac{W_D}{1-\pi_{\mathrm{tie}}}\right|<1\), and \(|W_D|<1\)
\end{enumerate}

Let $F_a$ be the marginal outcome law defined in \ref{app:equiv}. For a generic treated outcome $\bm y$ and control outcome $\bm z$, define
\[
\varphi_1^D(\bm y)=\int\mathcal W(\bm y,\bm z)\,F_0(d\bm z)-W_D,
\qquad
\varphi_0^D(\bm z)=\int\mathcal W(\bm y,\bm z)\,F_1(d\bm y)-W_D,
\]
and
\[
G_{1i}^D=\sum_{j=1}^{N_i}\varphi_1^D(\bm Y_{ij}),\quad A_i=1,
\qquad
G_{0k}^D=\sum_{l=1}^{N_k}\varphi_0^D(\bm Y_{kl}),\quad A_k=0.
\]
These are mean zero functions of individual cluster records. In particular, the control projection uses the treated versus control orientation of $\mathcal W$ and therefore enters the linearization with a plus sign. Put $\nu_a^2=\operatorname{Var}(G_{ai}^D\mid A_i=a)$ and
\[
V_{W_D}=\frac{\nu_1^2}{q\mu_1^2}+\frac{\nu_0^2}{(1-q)\mu_0^2},
\qquad
v_{W_D,M}=\frac{\nu_1^2}{M_1\mu_1^2}+\frac{\nu_0^2}{M_0\mu_0^2}.
\]
Here $V_{W_D}$ is the variance of the $\sqrt M$ limit, whereas $v_{W_D,M}$ is the first order ordinary sampling variance. They satisfy $Mv_{W_D,M}\to V_{W_D}$.

\begin{theorem}
\label{thm:wd_clt_app}
Under (C1)--(C3),
\[
\widehat W_D-W_D=\frac{1}{M_1\mu_1}\sum_{i:A_i=1}G_{1i}^D +\frac{1}{M_0\mu_0}\sum_{k:A_k=0}G_{0k}^D +o_p(M^{-1/2}),
\]
and thus
\[
\sqrt M(\widehat W_D-W_D)\xrightarrow{d}N(0,V_{W_D}).
\]
\end{theorem}

\begin{proof}
For a treated cluster $i$ and a control cluster $k$, define
\[
H_{ik}^D=\sum_{j=1}^{N_i}\sum_{l=1}^{N_k}\mathcal W(\bm Y_{ij},\bm Y_{kl}),
\qquad
\widetilde H_{ik}^D=H_{ik}^D-W_DN_iN_k.
\]
Non-informative cluster size and between-cluster independence imply
\[
\mathbb E(\widetilde H_{ik}^D\mid\mathcal C_i)=\mu_0G_{1i}^D,
\qquad
\mathbb E(\widetilde H_{ik}^D\mid\mathcal C_k)=\mu_1G_{0k}^D,
\qquad \mathbb E(\widetilde H_{ik}^D)=0.
\]
Define the degenerate kernel
\[
D_{ik}^D=\widetilde H_{ik}^D-\mu_0G_{1i}^D-\mu_1G_{0k}^D.
\]
It has conditional mean zero given either cluster. Writing $U_D=(M_1M_0)^{-1}\sum_{i:A_i=1}\sum_{k:A_k=0}H_{ik}^D$, the following decomposition is exact,
\[
U_D-\overline N_1\overline N_0W_D=\frac{\mu_0}{M_1}\sum_{i:A_i=1}G_{1i}^D+\frac{\mu_1}{M_0}\sum_{k:A_k=0}G_{0k}^D+R_{D,M},
\quad
R_{D,M}=\frac{1}{M_1M_0}\sum_{i:A_i=1}\sum_{k:A_k=0}D_{ik}^D.
\]
For two different cluster pairs, the corresponding degenerate terms are uncorrelated. This follows from independence for disjoint pairs and from conditioning on the shared cluster otherwise. Moreover, $|\widetilde H_{ik}^D|\leq 2N_iN_k$, so (C3) gives $\mathbb E\{(D_{ik}^D)^2\}<\infty$. Therefore,
\[
\mathbb E(R_{D,M}^2)=\frac{\mathbb E\{(D_{ik}^D)^2\}}{M_1M_0}=O(M^{-2}),
\]
and $R_{D,M}=O_p(M^{-1})$. The law of large numbers gives $\overline N_a\to_p\mu_a>0$. Divide the exact decomposition by $\overline N_1\overline N_0$ and use $\widehat W_D=U_D/(\overline N_1\overline N_0)$. Replacing the convergent coefficients by $1/\mu_a$ contributes $o_p(M^{-1/2})$, which gives the linearization in main paper. Finally, the two sums consist of independent mean zero cluster contributions with finite second moments. Their central limit theorems, $M_1/M\to q$, and Slutsky's theorem give the asserted normal limit and $V_{W_D}$.
\end{proof}

Theorem \ref{thm:wd_clt_app} shows the asymptotic properties for the cluster rank-sum estimator. For the FCL variance estimator, we additionally suppose that $F_1=F_0=F$ and $\mu_1=\mu_0=\mu$, which are sufficient conditions under the null $W_D=0$. Define
\[
r(\bm y)=\int\mathcal W(\bm y,\bm z)F(d\bm z),\qquad
R_i=\sum_{j=1}^{N_i}r(\bm Y_{ij}),\qquad n=n_1+n_0.
\]
Antisymmetry gives $W_D=0$, $G_{1i}^D=R_i$, and $G_{0i}^D=-R_i$. Conditional on $\mathcal C_i$, comparisons with different partner clusters are independent, have conditional mean $\mu R_i$, and have conditional second moments bounded by a constant times $N_i^2$. The within-cluster contribution to $S_i$ cancels. Averaging the resulting conditional mean-square bounds and using $n/M\to_p\mu$ yields
\[
\frac{1}{M_a}\sum_{i:A_i=a}\left(\frac{S_i}{n}-R_i\right)^2=o_p(1),\qquad a\in\{0,1\}.
\]
Thus, with
\[
\overline S_a=\frac{1}{M_a}\sum_{i:A_i=a}S_i,\qquad s_{S,a}^2=\frac{1}{M_a-1}\sum_{i:A_i=a}(S_i-\overline S_a)^2,
\]
we have $s_{S,a}^2/n^2\to_p\nu_a^2$. Thus the ordinary FCL variance estimator
\[
\widehat v_{W_D,M}^{\mathrm{FCL}} =\left\{\frac{Mq_M(1-q_M)}{n_1n_0}\right\}^{\!2} \left(\frac{s_{S,1}^2}{M_1}+\frac{s_{S,0}^2}{M_0}\right)
\]
satisfies $M\widehat v_{W_D,M}^{\mathrm{FCL}}\to_p V_{W_D}$ under these additional conditions. 

For the other win measures, let
\[
g_R(w,t)=2\,\operatorname{atanh}\{w/(1-t)\},\qquad
g_O(w)=2\,\operatorname{atanh}(w),\qquad g_D(w)=(1+w)/2.
\]
Write $\bm\pi=(\pi_{\mathrm{win}},\pi_{\mathrm{loss}})^\top$ and
\[
B=\begin{pmatrix}1&-1\\-1&-1\end{pmatrix},\qquad
V_{D,t}=BV_\pi B^\top,
\]
where $V_\pi$ is defined in Theorem \ref{thm:u_joint}. It also establishes the joint result
\[
\sqrt M\begin{pmatrix}\widehat W_D-W_D\\\widehat\pi_{\mathrm{tie}}-\pi_{\mathrm{tie}}\end{pmatrix}
\xrightarrow{d}N(\bm0,V_{D,t}).
\]
For $d=(1-t)^2-w^2$, the derivatives are
\[
a(w,t)=\frac{2(1-t)}{d},\qquad b(w,t)=\frac{2w}{d},\qquad
g_O'(w)=\frac{2}{1-w^2},\qquad g_D'(w)=\frac12.
\]
Under (C4) and the conditions of Theorem \ref{thm:u_joint}, the delta method gives limiting variances
\[
V_{\log W_R}=(a,b)V_{D,t}(a,b)^\top,\qquad
V_{\log W_O}=\frac{4V_{W_D}}{(1-W_D^2)^2},\qquad
V_{\mathrm{DOOR}}=\frac14V_{W_D},
\]
with $a=a(W_D,\pi_{\mathrm{tie}})$ and $b=b(W_D,\pi_{\mathrm{tie}})$. Define the ordinary covariance estimator $\widehat v_{D,t,M}=B\widehat v_{\pi,M}B^\top$ using \ref{app:wald_u}, and denote its entries by $\widehat v_{W_D,M}$, $\widehat c_{D,t,M}$, and $\widehat v_{t,M}$. With $\widehat a=a(\widehat W_D,\widehat\pi_{\mathrm{tie}})$ and $\widehat b=b(\widehat W_D,\widehat\pi_{\mathrm{tie}})$, the full sampling variance estimator is
\[
\widehat v_{\log W_R,M}=\widehat a^2\widehat v_{W_D,M}+2\widehat a\widehat b\widehat c_{D,t,M}+\widehat b^2\widehat v_{t,M}.
\]
Similarly,
\[
\widehat v_{\log W_O,M}=\frac{4\widehat v_{W_D,M}}{(1-\widehat W_D^2)^2},\qquad
\widehat v_{\mathrm{DOOR},M}=\frac14\widehat v_{W_D,M}.
\]
Each satisfies $M\widehat v_{\tau,M}\to_p V_\tau$. At $W_D=0$, $b=0$, so tie probability variability contributes nothing at first order to the null variance of $\log(\widehat W_R)$. Under the additional FCL calibration conditions, $\widehat a^2\widehat v_{W_D,M}^{\mathrm{FCL}}$ is therefore a null consistent variance estimator.

\section{Regularity conditions and proofs for the bivariate clustered $U$-statistic Wald test (ZJ)}
\label{app:wald_u}

We derive the joint probability limit using centered cluster pair kernels. This extends the clustered two-sample $U$-statistic construction of \citet{zhang2021inference} while explicitly accounting for the random arm-specific mean cluster sizes in the estimators.

For $c\in\{\mathrm{win},\mathrm{loss}\}$, write
\[
H_{ik}^c=\sum_{j=1}^{N_i}\sum_{l=1}^{N_k}\phi_c(\bm Y_{ij},\bm Y_{kl}),\qquad
U_c=\frac{1}{M_1M_0}\sum_{i:A_i=1}\sum_{k:A_k=0}H_{ik}^c,
\qquad \widehat\pi_c=\frac{U_c}{\overline N_1\overline N_0}.
\]
The population mean of the raw count kernel is $\theta_c=\mu_1\mu_0\pi_c$, not $\overline N_1\overline N_0\pi_c$. Define the subject-level projections
\[
\varphi_1^c(\bm y)=\int\phi_c(\bm y,\bm z)F_0(d\bm z)-\pi_c,
\qquad
\varphi_0^c(\bm z)=\int\phi_c(\bm y,\bm z)F_1(d\bm y)-\pi_c,
\]
and their cluster sums
\[
G_{ai}^c=\sum_{j=1}^{N_i}\varphi_a^c(\bm Y_{ij}),\qquad
\bm G_{ai}=(G_{ai}^{\mathrm{win}},G_{ai}^{\mathrm{loss}})^\top,
\qquad A_i=a.
\]
These equal the quantities denoted by $\bm G_i$ in the main paper. Non-informative cluster size implies $\mathbb E(\bm G_{ai}\mid A_i=a)=\bm0$. In addition to (C1)--(C3), impose the following conditions. Condition (C4) is needed only for the ratio transforms.
\begin{enumerate}
\item[(C5)] The covariance matrices $\Sigma_a=\operatorname{Var}(\bm G_{ai}\mid A_i=a)$ are finite. This follows from (C3) for the bounded win and loss kernels, since $\|\bm G_{ai}\|\leq\sqrt2N_i$.
\item[(C6)] With
\[
V_\pi=\frac{\Sigma_1}{q\mu_1^2}+\frac{\Sigma_0}{(1-q)\mu_0^2},
\]
we have $\nabla g(\bm\pi)^\top V_\pi\nabla g(\bm\pi)>0$ for each estimand on which inference is performed. The matrix itself need not be positive definite.
\end{enumerate}

\begin{theorem}
\label{thm:u_joint}
Under (C1)--(C3), (C5), and (C6),
\[
\widehat{\bm\pi}-\bm\pi=\frac{1}{M_1\mu_1}\sum_{i:A_i=1}\bm G_{1i}+\frac{1}{M_0\mu_0}\sum_{k:A_k=0}\bm G_{0k} +o_p(M^{-1/2}),
\]
and
\[
\sqrt M(\widehat{\bm\pi}-\bm\pi)\xrightarrow{d}N(\bm0,V_\pi).
\]
The first-order ordinary sampling covariance is
\[
v_{\pi,M}=\frac{\Sigma_1}{M_1\mu_1^2}+\frac{\Sigma_0}{M_0\mu_0^2},
\qquad Mv_{\pi,M}\longrightarrow V_\pi.
\]
\end{theorem}

\begin{proof}
For either component, the treatment arm projection of the uncentered raw count kernel satisfies
\[
\mathbb E(H_{ik}^c\mid\mathcal C_i)-\theta_c =\mu_0G_{1i}^c+\mu_0\pi_c(N_i-\mu_1).
\]
The control-arm projection analogously contains $\mu_1\pi_c(N_k-\mu_0)$. Consider \(\widetilde H_{ik}^c=H_{ik}^c-\pi_cN_iN_k\).  To obtain its treatment-arm projection, integrate first over a control cluster conditional on its size. Non-informative cluster size gives
\begin{align*}
\mathbb E(H_{ik}^c\mid\mathcal C_i)
&=\mu_0\sum_{j=1}^{N_i}\int\phi_c(\bm Y_{ij},\bm z)F_0(d\bm z),\\
\mathbb E(\pi_cN_iN_k\mid\mathcal C_i)&=\mu_0\pi_cN_i.
\end{align*}
Subtracting these two expressions yields
\[
\mathbb E(\widetilde H_{ik}^c\mid\mathcal C_i)=\mu_0G_{1i}^c,
\qquad
\mathbb E(\widetilde H_{ik}^c\mid\mathcal C_k)=\mu_1G_{0k}^c.
\]
Both projections have mean zero. Let
\[
D_{ik}^c=\widetilde H_{ik}^c-\mu_0G_{1i}^c-\mu_1G_{0k}^c.
\]
Its conditional expectation given either cluster is zero. Since $|\widetilde H_{ik}^c|\leq N_iN_k$, (C3) implies finite second moments. Summing the centered kernels gives the exact identity
\[
U_c-\overline N_1\overline N_0\pi_c
=\frac{\mu_0}{M_1}\sum_{i:A_i=1}G_{1i}^c
 +\frac{\mu_1}{M_0}\sum_{k:A_k=0}G_{0k}^c+R_{c,M},
\]
where $R_{c,M}=(M_1M_0)^{-1}\sum_{i:A_i=1}\sum_{k:A_k=0}D_{ik}^c$. Put $\bm R_M=(R_{\mathrm{win},M},R_{\mathrm{loss},M})^\top$ and define $\bm D_{ik}$ correspondingly. Degeneracy and conditional independence imply
\[
\mathbb E\|\bm R_M\|^2=\frac{\mathbb E\|\bm D_{ik}\|^2}{M_1M_0}=O(M^{-2}),
\]
so $\bm R_M=O_p(M^{-1})$. In vector form,
\[
\bm U-\overline N_1\overline N_0\bm\pi=\frac{\mu_0}{M_1}\sum_{i:A_i=1}\bm G_{1i} +\frac{\mu_1}{M_0}\sum_{k:A_k=0}\bm G_{0k}+\bm R_M.
\]
The two projected sums are independent conditional on the allocation and have covariance matrices $\Sigma_1/M_1$ and $\Sigma_0/M_0$. Hence their joint central limit theorem gives
\[
\sqrt M\{\bm U-\overline N_1\overline N_0\bm\pi\}
\xrightarrow{d}N\!\left(\bm0,
\frac{\mu_0^2}{q}\Sigma_1+\frac{\mu_1^2}{1-q}\Sigma_0\right).
\]
Finally, $\overline N_1\overline N_0\to_p\mu_1\mu_0>0$. Dividing the exact decomposition by this denominator and applying Slutsky's theorem proves the covariance $V_\pi$.
\end{proof}

For implementation, the empirical projected cluster vectors can be written directly using cluster-pair counts,
\[
\widehat G_{1i}^c=\frac{1}{n_0}\sum_{k:A_k=0}H_{ik}^c-N_i\widehat\pi_c,
\qquad
\widehat G_{0k}^c=\frac{1}{n_1}\sum_{i:A_i=1}H_{ik}^c-N_k\widehat\pi_c.
\]
These are exactly the within-cluster sums of the empirically centered individual-level projections in the main paper. Let $\widehat{\bm G}_{ai}$ collect wins and losses and put
\[
\overline{\bm G}_a=\frac{1}{M_a}\sum_{i:A_i=a}\widehat{\bm G}_{ai},\qquad
\widehat\Sigma_a=\frac{1}{M_a-1}\sum_{i:A_i=a}
(\widehat{\bm G}_{ai}-\overline{\bm G}_a)(\widehat{\bm G}_{ai}-\overline{\bm G}_a)^\top.
\]
In exact arithmetic, $\overline{\bm G}_a=\bm0$. The finite-sample covariance estimator remains
\[
\widehat v_{\pi,M}=\frac{\widehat\Sigma_1}{M_1\overline N_1^2}+\frac{\widehat\Sigma_0}{M_0\overline N_0^2}.
\]
Equivalently,
\[
\widehat\Sigma_U=\frac{\overline N_0^2}{M_1}\widehat\Sigma_1+\frac{\overline N_1^2}{M_0}\widehat\Sigma_0,
\qquad
\widehat\Sigma_\pi=\frac{\widehat\Sigma_U}{(\overline N_1\overline N_0)^2} =\widehat v_{\pi,M},
\]
where $\widehat\Sigma_\pi$ is the ordinary covariance in the main paper.

Consider the treatment-arm projection for component $c$, and define
\[
B_{1i}^c=\mathbb E(H_{ik}^c\mid\mathcal C_i),\qquad
E_{1i}^c=\frac{1}{M_0}\sum_{k:A_k=0}H_{ik}^c-B_{1i}^c.
\]
Conditional independence and bounded kernels give
\[
\mathbb E\{(E_{1i}^c)^2\mid\mathcal C_i\}\leq\frac{N_i^2\mathbb E(N_k^2\mid A_k=0)}{M_0}.
\]
Since $|B_{1i}^c|\leq\mu_0N_i$, the identity
\[
\widehat G_{1i}^c-G_{1i}^c=\frac{E_{1i}^c}{\overline N_0}+\left(\frac{1}{\overline N_0}-\frac{1}{\mu_0}\right)B_{1i}^c-N_i(\widehat\pi_c-\pi_c)
\]
implies
\[
\frac{1}{M_1}\sum_{i:A_i=1}(\widehat G_{1i}^c-G_{1i}^c)^2=o_p(1).
\]
The first term follows by the conditional mean square bound and Markov's inequality, and the remaining terms follow from consistency of $\overline N_0$ and $\widehat\pi_c$ and the law of large numbers for $N_i^2$. The control-arm argument is identical. The laws of large numbers for $\bm G_{ai}$ and $\bm G_{ai}\bm G_{ai}^\top$, together with Cauchy--Schwarz, then give $\widehat\Sigma_a\to_p\Sigma_a$. It follows that
\[
M\widehat v_{\pi,M}\xrightarrow{p}V_\pi.
\]
This accounts for the dependence among the empirical projected vectors through their mean square approximation.

For a differentiable $\tau=g(\bm\pi)$, define
\[
V_\tau=\nabla g(\bm\pi)^\top V_\pi\nabla g(\bm\pi),\qquad
\widehat v_{\tau,M}=\nabla g(\widehat{\bm\pi})^\top\widehat v_{\pi,M}\nabla g(\widehat{\bm\pi}).
\]
The delta method gives
\[
\sqrt M(\widehat\tau-\tau)\xrightarrow{d}N(0,V_\tau),\qquad
M\widehat v_{\tau,M}\xrightarrow{p}V_\tau.
\]
Thus $(\widehat\tau-\tau_0)/\sqrt{\widehat v_{\tau,M}}\to_d N(0,1)$ under the corresponding null. A $t$ reference with $M-2$ degrees of freedom is a finite-sample calibration choice. The gradients, evaluated at the population probabilities, are
\begin{align*}
\nabla W_D&=(1,-1)^\top,\\
\nabla\log W_R&=(1/\pi_{\mathrm{win}},-1/\pi_{\mathrm{loss}})^\top,\\
\nabla\log W_O&=\frac{2}{1-W_D^2}(1,-1)^\top,\\
\nabla\mathrm{DOOR}&=\tfrac12(1,-1)^\top.
\end{align*}
Condition (C4) ensures that the ratio derivatives are finite and continuous near the truth. The joint limit for $(\widehat W_D,\widehat\pi_{\mathrm{tie}})$ in \ref{app:wald} follows by applying the matrix $B$ to Theorem \ref{thm:u_joint}. In particular, the signed and bivariate representations yield the same first order sampling variance for the common estimator $\widehat W_D$.

\section{Randomization-based tests and implementation}
\label{app:score}

We derive the finite-population variance of the cluster FS score under complete randomization and distinguish its normal approximation from Monte Carlo randomization tests. In this section. we use a sharp null under which the complete observed cluster records entering the comparison rule are unaffected by assignment. This includes the cluster sizes and observation or censoring features needed to keep the comparison scores fixed. 

Let $\mathcal A=\{\bm a\in\{0,1\}^M:\sum_i a_i=M_1\}$ and $q_M=M_1/M$. Under the sharp null, the outcomes and pairwise comparison scores remain fixed across assignments, while the test statistic is recomputed for each assignment. Conditional on these fixed records, write
\[
S_i=\sum_{j=1}^{N_i}\sum_{k=1}^M\sum_{l=1}^{N_k}s_{ij,kl},\qquad
\sum_{i=1}^M S_i=0,\qquad \overline S=M^{-1}\sum_iS_i.
\]
For $\bm A$ uniform on $\mathcal A$,
\[
\mathbb E(A_i)=q_M,\qquad \operatorname{Var}(A_i)=q_M(1-q_M),\qquad
\operatorname{Cov}(A_i,A_{i'})=-\frac{q_M(1-q_M)}{M-1},\quad i\ne i'.
\]
Define $T_S(\bm A)=\sum_i(A_i-q_M)S_i$. Its conditional expectation is zero and its exact conditional variance is
\[
v_{S,M}=\operatorname{Var}\{T_S(\bm A)\mid\bm S\}
=\frac{q_M(1-q_M)M}{M-1}\sum_{i=1}^M(S_i-\overline S)^2.
\]

\begin{proof}
Treating the scores as fixed gives
\begin{align*}
\operatorname{Var}\{T_S(\bm A)\mid\bm S\}
&=q_M(1-q_M)\sum_iS_i^2-\frac{2q_M(1-q_M)}{M-1}\sum_{i<i'}S_iS_{i'}\\
&=\frac{q_M(1-q_M)M}{M-1}\left(\sum_iS_i^2-M\overline S^{\,2}\right),
\end{align*}
where $2\sum_{i<i'}S_iS_{i'}=(\sum_iS_i)^2-\sum_iS_i^2$ was used in the second equality. This is the asserted variance.
\end{proof}

The analytic cluster FS statistic is
\[
Z_{\mathrm{FS}}=\frac{T_S(\bm A)}{\sqrt{v_{S,M}}}.
\]
For a sequence of finite populations, suppose $q_M\to q\in(0,1)$, $\sum_i(S_i-\overline S)^2>0$, and
\[
\frac{\max_i(S_i-\overline S)^2}{\sum_i(S_i-\overline S)^2}\longrightarrow0.
\]
The finite population central limit theorem then gives $Z_{\mathrm{FS}}\to_d N(0,1)$ under the randomization distribution \citep{li2017general}. The normal reference FS test is asymptotic, even though its conditional variance is exact. A $t$ reference is an additional finite-sample approximation. 

For a prespecified statistic $\widehat\tau(\bm a)$ with null center $\tau_0$, a two-sided Monte Carlo randomization test using $B$ independently and uniformly sampled assignments $\bm a^{(1)},\ldots,\bm a^{(B)}\in\mathcal A$ is
\[
\widehat p_\tau=\frac{1+\sum_{b=1}^B\mathbb I\{|\widehat\tau(\bm a^{(b)})-\tau_0|\geq |\widehat\tau(\bm A)-\tau_0|\}}{B+1}.
\]
The observed assignment and the sampled assignments are exchangeable under the sharp null. Thus the inclusive rank and the addition of one produce a valid Monte Carlo randomization test, with possible conservativeness from discreteness and ties. Exhaustive enumeration gives the corresponding exact randomization distribution. Both require the statistic and its ordering rule to be specified for every assignment considered.

When the statistic is a normalized win measure, each assignment requires
\[
n_1(\bm a)=\sum_i a_iN_i,\qquad n_0(\bm a)=\sum_i(1-a_i)N_i,
\qquad
\widehat W_D(\bm a)=\frac{T_S(\bm a)}{n_1(\bm a)n_0(\bm a)},
\]
together with the win, loss, and tie counts for that assignment. A test ordering assignments by $|T_S(\bm a)|$ is not generally the same test as one ordering them by $|\widehat W_D(\bm a)|$ when cluster sizes are unequal.

For log-scale ratio tests, use null center zero for $\log(\widehat W_R)$ and $\log(\widehat W_O)$, zero for $\widehat W_D$, and $1/2$ for $\widehat{\mathrm{DOOR}}$. At every assignment where the expressions are defined,
\[
|\widehat{\mathrm{DOOR}}(\bm a)-1/2|=\tfrac12|\widehat W_D(\bm a)|,
\qquad
|\log\widehat W_O(\bm a)|=2\operatorname{atanh}\{|\widehat W_D(\bm a)|\}.
\]
The transformations on the right are strictly increasing in $|\widehat W_D(\bm a)|$. Therefore the three procedures have identical extremeness indicators and identical Monte Carlo $p$-values when they use the same assignments.

A native-scale permutation test for a ratio instead uses $|\widehat W_R(\bm a)-1|$ or $|\widehat W_O(\bm a)-1|$, which is a different two-sided ordering from absolute log deviation.

\section{Regularity conditions and asymptotic theory for the cluster-level JEL test}
\label{app:jel}

We establish the jackknife empirical likelihood ratio result through approximation by independent cluster influence contributions. The construction follows the empirical likelihood and JEL framework of \citet{qin1994empirical,jing2009jackknife,peng2018jackknife}, and Peng and Tan. The observed pseudo-values are not independent, because their leave-one-cluster samples overlap.

For a scalar estimator $\widehat\tau$, define
\[
\widehat\tau_i=M\widehat\tau-(M-1)\widehat\tau_{(-i)},\qquad
\overline{\widehat\tau}_{\mathrm{PV}}=\frac1M\sum_i\widehat\tau_i,
\qquad Z_i(t)=\widehat\tau_i-t.
\]
For a candidate value $t$, let
\[
\mathcal L(t)=\sup\left\{\prod_{i=1}^Mp_i:p_i\geq0,\ \sum_i p_i=1,\ \sum_i p_iZ_i(t)=0\right\}.
\]
The exact pseudo-value identity is
\[
\overline{\widehat\tau}_{\mathrm{PV}}=M\widehat\tau-(M-1)\overline{\widehat\tau}_{(-)},
\qquad \overline{\widehat\tau}_{(-)}=\frac1M\sum_i\widehat\tau_{(-i)}.
\]
The pseudo-value mean need not equal $\widehat\tau$. By the arithmetic-geometric mean inequality, the unrestricted maximum is $M^{-M}$, attained at equal weights and $t=\overline{\widehat\tau}_{\mathrm{PV}}$. Thus the likelihood ratio is
\[
R(t)=-2\log\{\mathcal L(t)/M^{-M}\}.
\]

In addition to (C1)--(C2), assume the following conditions for the chosen scalar target.
\begin{enumerate}
\item[(C7)] There are independent, mean-zero cluster influence contributions $\xi_{i,M}$ such that
\[
\widehat\tau-\tau=\frac1M\sum_i\xi_{i,M}+o_p(M^{-1/2}),\qquad
\frac1M\sum_i\operatorname{Var}(\xi_{i,M})\longrightarrow V_\tau\in(0,\infty).
\]
Their variances may differ between arms.
\item[(C8)] For some $\delta>0$, $\sup_{M,i}\mathbb E|\xi_{i,M}|^{2+\delta}<\infty$. In particular the Lindeberg condition holds and $M^{-1}\sum_i\xi_{i,M}^2\to_p V_\tau$.
\item[(C9)] The leave-one-cluster estimates are uniformly consistent and the pseudo-values satisfy both the mean and mean-square approximations
\[
\max_i|\widehat\tau_{(-i)}-\tau|=o_p(1),\qquad
\overline{\widehat\tau}_{\mathrm{PV}}-\widehat\tau=o_p(M^{-1/2}),
\]
\[
\frac1M\sum_i(\widehat\tau_i-\tau-\xi_{i,M})^2=o_p(1).
\]
\item[(C10)] Under $H_0:\tau=\tau_0$, the set $\{Z_i(\tau_0)\}$ contains both positive and negative values with probability tending to one.
\end{enumerate}

\begin{theorem}
\label{thm:jel_crt_wilks}
Under (C1), (C2), and (C7)--(C10),
\[
R(\tau_0)\xrightarrow{d}\chi_1^2\qquad\text{under }H_0:\tau=\tau_0.
\]
Thus rejection when $R(\tau_0)>\chi^2_{1,1-\alpha}$ has asymptotic level $\alpha$.
\end{theorem}

\begin{proof}
Write $Z_i=Z_i(\tau_0)$, $\overline Z=M^{-1}\sum_iZ_i$, $Q_M=M^{-1}\sum_iZ_i^2$, and $Z_{\max}=\max_i|Z_i|$. Conditions (C7)--(C9) imply
\[
\sqrt M\,\overline Z\xrightarrow{d}N(0,V_\tau),\qquad Q_M\xrightarrow{p}V_\tau.
\]
For the second assertion, apply Cauchy--Schwarz to the difference between the empirical second moments of $Z_i$ and $\xi_{i,M}$. The moment bound in (C8) gives
\[
\mathbb P\{\max_i|\xi_{i,M}|>\epsilon\sqrt M\}\leq\frac{\sum_i\mathbb E|\xi_{i,M}|^{2+\delta}} {\epsilon^{2+\delta}M^{1+\delta/2}}\longrightarrow0.
\]
Moreover, (C9) gives
\[
\frac{\max_i|Z_i-\xi_{i,M}|}{\sqrt M} \leq\left\{\frac1M\sum_i(Z_i-\xi_{i,M})^2\right\}^{1/2}=o_p(1).
\]
Therefore $Z_{\max}=o_p(\sqrt M)$.

On the condition in (C10), the Lagrange equation
\[
f_M(\lambda)=\frac1M\sum_i\frac{Z_i}{1+\lambda Z_i}=0
\]
is strictly decreasing on the interval where all denominators are positive and has a unique root there. Its sign equals the sign of $\overline Z$. The exact identity
\[
\overline Z=\lambda\frac1M\sum_i\frac{Z_i^2}{1+\lambda Z_i}
\]
implies
\[
|\lambda|Q_M\leq|\overline Z|(1+|\lambda|Z_{\max}).
\]
Since $|\overline Z|Z_{\max}=o_p(1)$ and $Q_M\to_p V_\tau>0$, it follows that
\[
|\lambda|\leq\frac{|\overline Z|}{Q_M-|\overline Z|Z_{\max}}=O_p(M^{-1/2}).
\]
Thus $\max_i|\lambda Z_i|=o_p(1)$ and
\[
\lambda=\frac{\overline Z}{Q_M}+o_p(M^{-1/2}).
\]
The constrained maximizing weights are $p_i=\{M(1+\lambda Z_i)\}^{-1}$, so
\[
R(\tau_0)=2\sum_i\log(1+\lambda Z_i).
\]
The third-order Taylor remainder is bounded by a constant times $\max_i|\lambda Z_i|\,\lambda^2\sum_iZ_i^2=o_p(1)$. Therefore
\[
R(\tau_0)=2M\lambda\overline Z-M\lambda^2Q_M+o_p(1) =\frac{M\overline Z^2}{Q_M}+o_p(1)\xrightarrow{d}\chi_1^2.
\]
\end{proof}

Let $\widehat\tau=g(\widehat{\bm\pi})$, where $g$ is twice continuously differentiable near the true probabilities, as holds for the four specified scales under their interior conditions. Theorem \ref{thm:u_joint} gives (C7) with \(\xi_{i,M}=\frac{\nabla g(\bm\pi)^\top\bm G_{ai}}{q_{a,M}\mu_a}\) with \(A_i=1\), \(q_{1,M}=q_M\), and  \(q_{0,M}=1-q_M\). Condition (C3) and $\|\bm G_{ai}\|\leq\sqrt2N_i$ give (C8).

For a cluster in arm $a$, direct deletion of its cross-arm counts yields the exact probability vector identity
\[
\bm\Delta_i:=\widehat{\bm\pi}-\widehat{\bm\pi}_{(-i)} =\frac{\widehat{\bm G}_{ai}}{n_a-N_i}.
\]
Here both arm totals and all cross-arm counts are recomputed after deletion. The empirical projections satisfy
$\sum_{i:A_i=a}\widehat{\bm G}_{ai}=\bm0$. Consequently,
\[\sum_{i=1}^M\bm\Delta_i=\sum_{a=0}^1\sum_{i:A_i=a} \frac{N_i\widehat{\bm G}_{ai}}{n_a(n_a-N_i)}.
\]
The moment condition gives $\sum_iN_i^2=O_p(M)$ and $\max_iN_i=o_p(\sqrt M)$, whereas $n_a/M\to_p q_a\mu_a>0$. Also $\|\widehat{\bm G}_{ai}\|\leq\sqrt2N_i$. Uniformly with probability tending to one, therefore,
\[
\left\|\sum_i\bm\Delta_i\right\|=O_p(M^{-1}),\qquad
\sum_i\|\bm\Delta_i\|^2=O_p(M^{-1}),\qquad
\max_i\|\bm\Delta_i\|=o_p(1).
\]
This proves uniform consistency of the leave-one-cluster probability estimates and keeps all required transforms in a neighborhood where their derivatives are bounded. Taylor expansion gives
\[
\widehat\tau-\widehat\tau_{(-i)} =\nabla g(\widehat{\bm\pi})^\top\bm\Delta_i+O(\|\bm\Delta_i\|^2)
\]
uniformly. Since
\[
\overline{\widehat\tau}_{\mathrm{PV}}-\widehat\tau =\frac{M-1}{M}\sum_i(\widehat\tau-\widehat\tau_{(-i)}),
\]
the two sum bounds imply
\[
\overline{\widehat\tau}_{\mathrm{PV}}-\widehat\tau=O_p(M^{-1})=o_p(M^{-1/2}).
\]
This shows the required mean approximation for the probability ratio estimators and their smooth transforms.

For the mean square approximation, write
\[
\widehat\tau_i-\tau=(\widehat\tau-\tau)+\frac{M-1}{n_a-N_i}\nabla g(\widehat{\bm\pi})^\top\widehat{\bm G}_{ai}+r_{i,M}.
\]
The Taylor bound gives $|r_{i,M}|\leq C N_i^2/M$ with probability tending to one. Hence
\[
\frac1M\sum_i r_{i,M}^2\leq\frac{C}{M^3}\sum_iN_i^4\leq\frac{C\max_iN_i^2}{M^3}\sum_iN_i^2=o_p(1).
\]
Uniformly within each arm,
$(M-1)/(n_a-N_i)-(q_{a,M}\mu_a)^{-1}=o_p(1)$. Combine this with consistency of $\nabla g(\widehat{\bm\pi})$, the meanv square approximation for $\widehat{\bm G}_{ai}$ proved in \ref{app:wald_u}, and $M^{-1}\sum_i\|\bm G_{ai}\|^2=O_p(1)$. The result is
\[
\frac1M\sum_i(\widehat\tau_i-\tau-\xi_{i,M})^2=o_p(1),
\]
which completes verification of (C9). At least one arm has a nondegenerate mean-zero influence distribution by (C6). It therefore has positive probability above $2\epsilon$ and below $-2\epsilon$ for some $\epsilon>0$. The within-arm laws of large numbers give order-$M$ counts in both sets. The number of indices for which $|\widehat\tau_i-\tau-\xi_{i,M}|\geq\epsilon$ is bounded by $\epsilon^{-2}\sum_i(\widehat\tau_i-\tau-\xi_{i,M})^2=o_p(M)$. Thus positive and negative pseudo-value deviations remain with probability tending to one, proving (C10) at the true value.

For ratio-type estimands, inference is performed on the log scale and results are back-transformed for reporting. Positivity of the full-sample and leave-one-cluster estimates, together with feasibility of the empirical likelihood constraint, must be checked. Logging does not by itself ensure these properties. The pseudo-values for a log target are formed from the log full-sample and log leave-one-cluster estimates For observed data, compute the full and leave-one-cluster estimates on the selected inferential scale, construct $Z_i(\tau_0)$, and check that all estimates are finite. If both signs occur, solve the Lagrange equation over
\[
-\frac{1}{\max_{Z_i>0}Z_i}<\lambda< -\frac{1}{\min_{Z_i<0}Z_i},
\]
and report $1-F_{\chi_1^2}\{R(\tau_0)\}$. If zero lies outside the convex hull, the constraint is infeasible. If zero is on its boundary but not every $Z_i$ is zero, every feasible product has a zero factor and the empirical likelihood ratio is infinite. If all $Z_i=0$, the empirical likelihood ratio is instead zero, although the nondegenerate Wilks approximation does not apply. Confidence sets on the selected scale are obtained by inversion,
\[
\{t:R(t)\leq\chi^2_{1,1-\alpha}\}.
\]

\section{Additional simulation results}
Figure \ref{fig:sim_typei_20_z}-\ref{fig:sim_power_100_t} present the empirical type I error and power results under the same simulation settings as in the main text, with the difference that all Wald-type procedures are referenced to the standard normal distribution.

\begin{sidewaysfigure}
    \centering
    \includegraphics[width=0.9\linewidth]{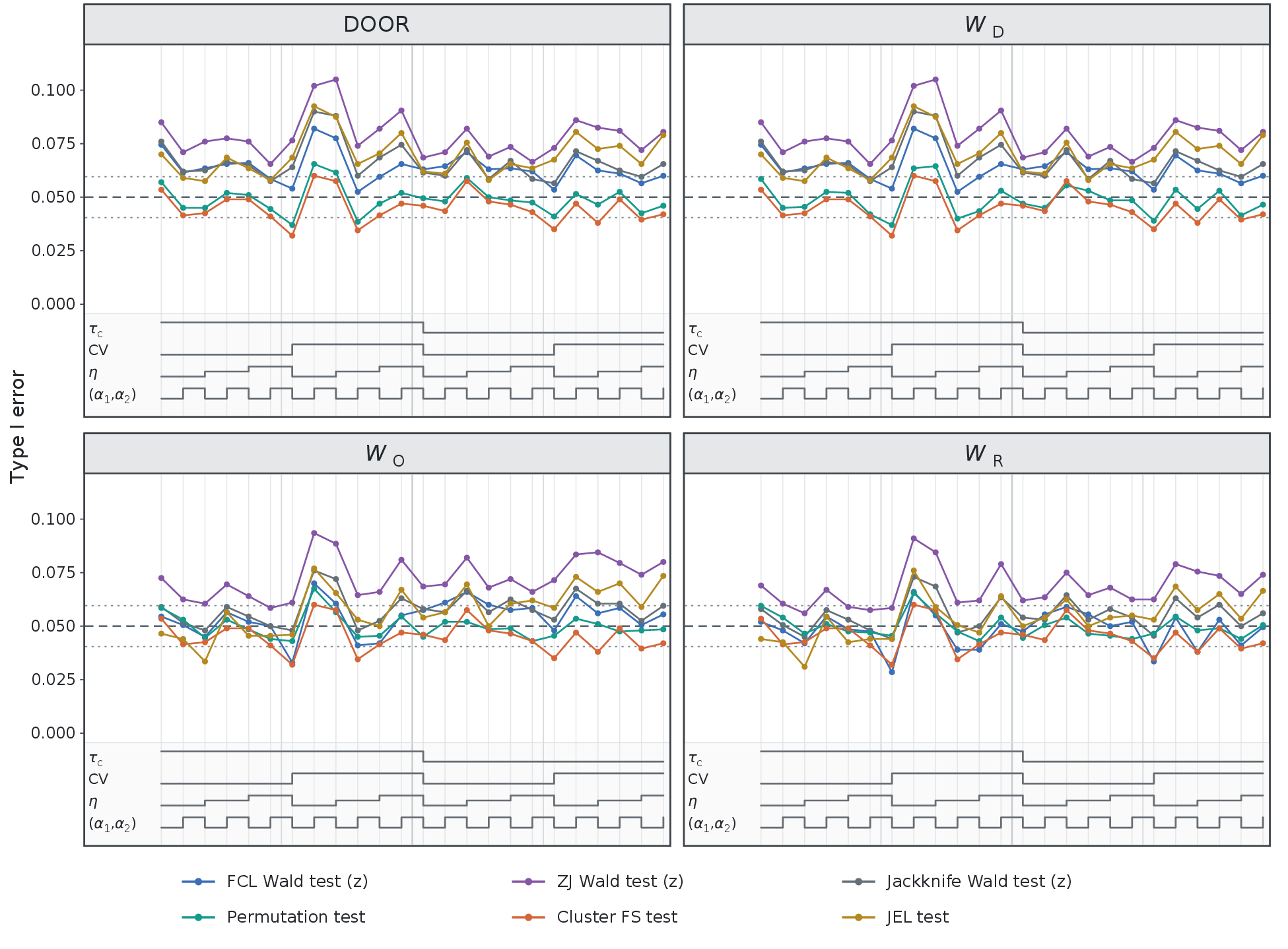}\\
    \caption{Empirical type I error for tests of win statistics in parallel-arm cluster-randomized trials with \(M=20\) under the global null \((\theta_1,\theta_2)=(0,0)\) across Monte Carlo 2,000 replicates with \(M=20\) clusters. Within each panel, results are shown for the four win statistics: \(W_D\), \(W_R\), \(W_O\), and \(\mathrm{DOOR}\). The procedure includes Wald z test proposed by \citet{fang2025sample}; Wald z test proposed by \citet{zhang2021inference}; Wald z test with delete-one-cluster jackknife standard errors; permutation test; cluster-level FS test \citep{finkelstein1999combining}; and the jackknife empirical likelihood ratio (JEL) test. The horizontal dashed line marks the nominal two-sided level \(\alpha=0.05\). The two horizontal dotted lines indicate the Monte Carlo variance band \(0.05 \pm 1.96\sqrt{0.05(1-0.05)/2000}\). The annotation strip above the encodes the scenario factors by: \(\pi_{\text{tie}}\) (35\% versus \(7\%\)), \(\mathrm{CV}(N_i)\) (0.3 versus 0.5 ), within-individual dependence parameter \(\eta\) (1, 2, versus 4), and frailty shape parameters \((\alpha_1,\alpha_2)\) (\((1,1)\) versus (2,2) ).}
    \label{fig:sim_typei_20_z} 
\end{sidewaysfigure}

\begin{sidewaysfigure}
    \centering
    \includegraphics[width=0.9\linewidth]{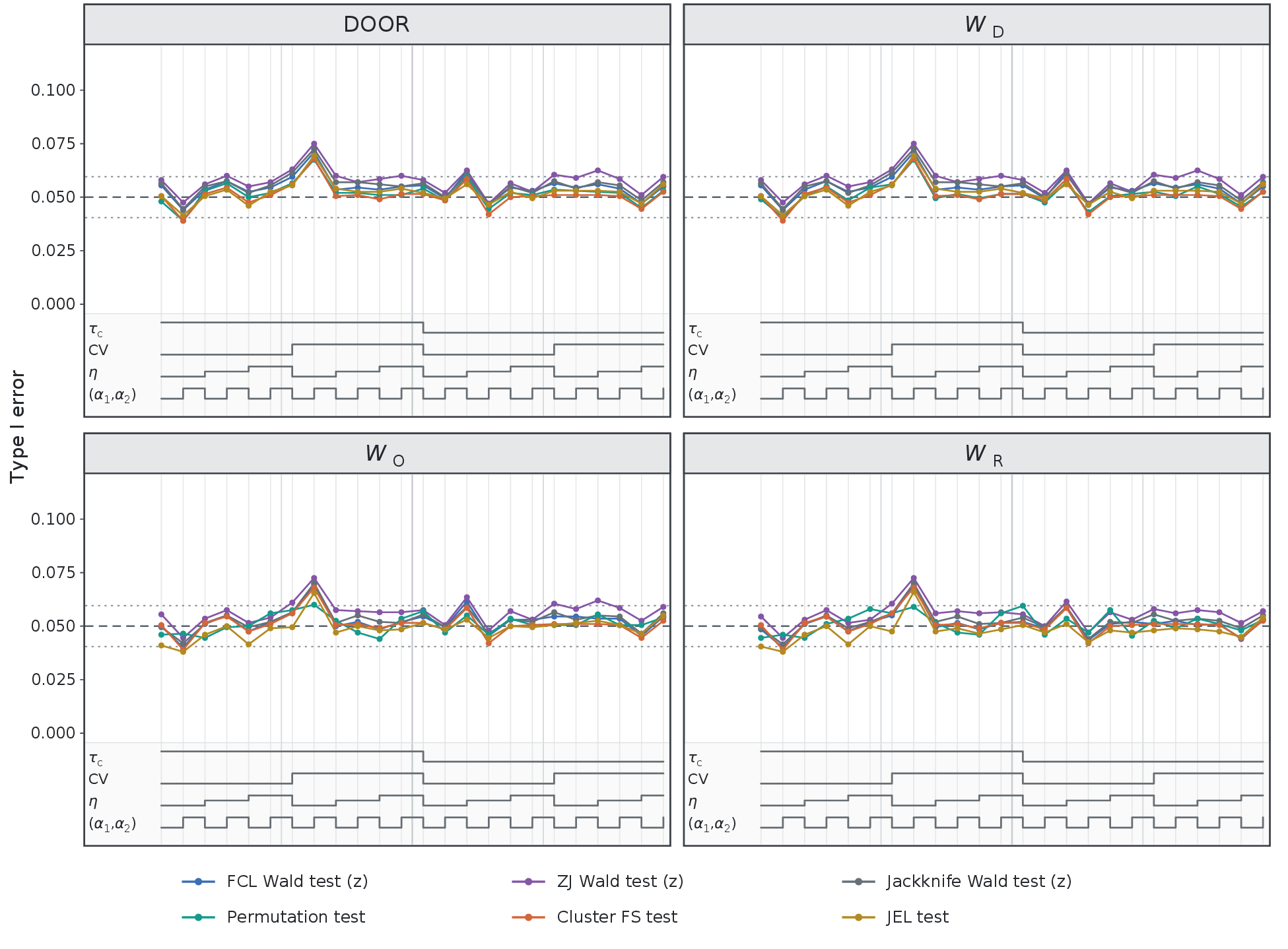}\\
    \caption{Empirical type I error for tests of win statistics in parallel-arm cluster-randomized trials with \(M=100\) under the global null \((\theta_1,\theta_2)=(0,0)\) across Monte Carlo 2,000 replicates with \(M=20\) clusters. Within each panel, results are shown for the four win statistics: \(W_D\), \(W_R\), \(W_O\), and \(\mathrm{DOOR}\). The procedure includes Wald t test with df\(=18\) using FCL \citep{fang2025sample}; Wald t test with df\(=18\)using ZJ \citep{zhang2021inference}; Wald t test with df\(=18\) delete-one-cluster jackknife standard errors; permutation test; cluster FS test \citep{finkelstein1999combining}; and the jackknife empirical likelihood (JEL) test. The horizontal dashed line marks the nominal two-sided level \(\alpha=0.05\). The two horizontal dotted lines indicate the Monte Carlo variance band \(0.05 \pm 1.96\sqrt{0.05(1-0.05)/2000}\). The annotation strip above the encodes the scenario factors by: \(\pi_{\text{tie}}\) (35\% versus \(7\%\)), \(\mathrm{CV}(N_i)\) (0.3 versus 0.5 ), within-individual dependence parameter \(\eta\) (1, 2, versus 4), and frailty shape parameters \((\alpha_1,\alpha_2)\) (\((1,1)\) versus (2,2) ).}
    \label{fig:sim_typei_100_t}
\end{sidewaysfigure}

\begin{sidewaysfigure}
    \centering
    \includegraphics[width=0.85\linewidth]{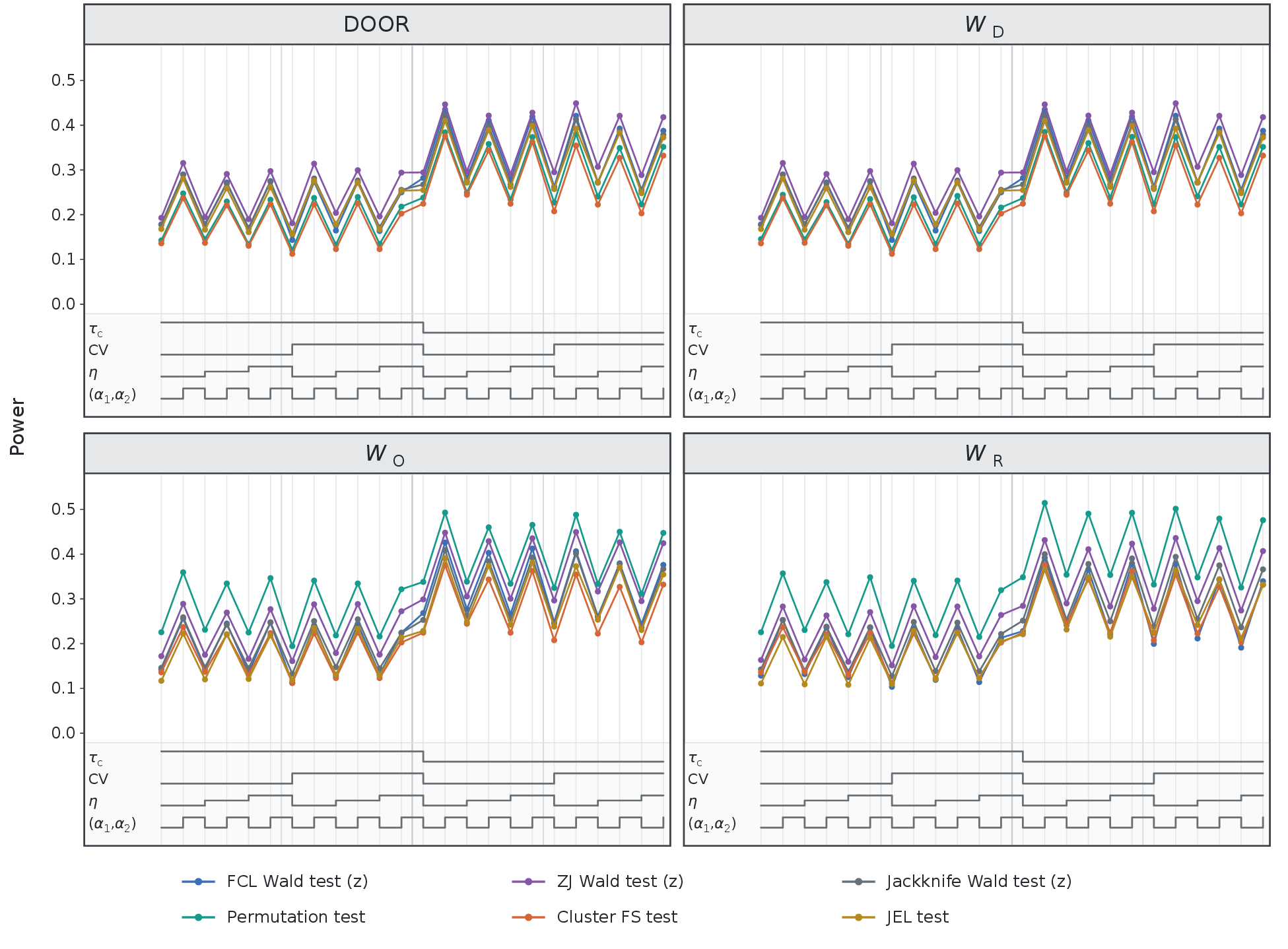}\\
    \caption{Empirical power for tests of win statistics in parallel-arm cluster-randomized trials with \(M=20\) under the concordant beneficial alternative \((\theta_1,\theta_2)=\{\log(0.65),\log(0.50)\}\) across 2{,}000 Monte Carlo replicates. The procedures include the Wald z test proposed by \citet{fang2025sample}; the Wald z test proposed by \citet{zhang2021inference}; a Wald z test with delete-one-cluster jackknife standard errors; the permutation test; the cluster FS test \citep{finkelstein1999combining}; and the jackknife empirical likelihood (JEL) test. Power is computed using a two-sided rejection criterion at the nominal level \(\alpha=0.05\). The annotation strip above the plot encodes the scenario factors by: \(\pi_{\text{tie}}\) (41\% versus 8\%), \(\mathrm{CV}(N_i)\) (0.3 versus 0.5), within-individual dependence parameter \(\eta\) (1, 2, versus 4), and frailty shape parameters \((\alpha_1,\alpha_2)\) (\((1,1)\) versus \((2,2)\)).}
    \label{fig:sim_power_20_z}
\end{sidewaysfigure}

\begin{sidewaysfigure}
    \centering
    \includegraphics[width=0.9\linewidth]{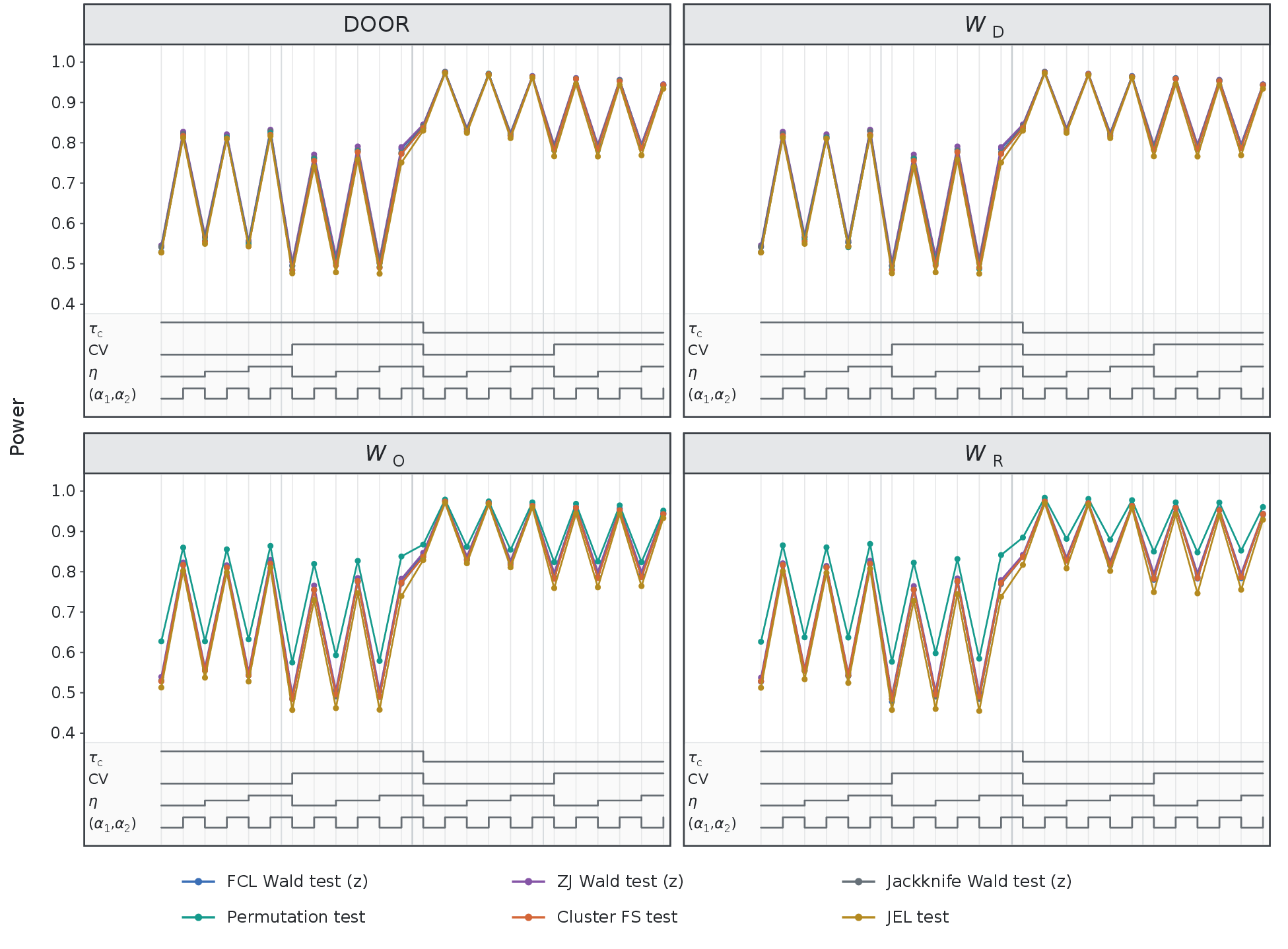}\\
    \caption{Empirical power for tests of win statistics in parallel-arm cluster-randomized trials with \(M=100\) under the concordant beneficial alternative \((\theta_1,\theta_2)=\{\log(0.65),\log(0.50)\}\) across 2{,}000 Monte Carlo replicates. The procedures include the Wald t test with df\(=18\) proposed using FCL \citep{fang2025sample}; the Wald t test with df\(=18\) using ZJ \citep{zhang2021inference}; a Wald t test with delete-one-cluster jackknife standard errors; exact permutation test; cluster FS test \citep{finkelstein1999combining}; and the jackknife empirical likelihood (JEL) test. Power is computed using a two-sided rejection criterion at the nominal level \(\alpha=0.05\). The annotation strip above the plot encodes the scenario factors by: \(\pi_{\text{tie}}\) (41\% versus 8\%), \(\mathrm{CV}(N_i)\) (0.3 versus 0.5), within-individual dependence parameter \(\eta\) (1, 2, versus 4), and frailty shape parameters \((\alpha_1,\alpha_2)\) (\((1,1)\) versus \((2,2)\)).}
    \label{fig:sim_power_100_t}
\end{sidewaysfigure}

\section{Illustrative data example with t-test}

Table \ref{supp:stride} reports the STRIDE data analysis results under the same inferential procedures as in Table 4 in the main text, with the Wald-type procedures referenced to a $t$ distribution with $84$ degrees of freedom.

\begin{table}[!ht] 
\centering
\caption{STRIDE illustration results for the four win statistics under the three Wald-type inferential procedures. Reported quantities include the point estimate (Est.), standard error (SE), and the two-sided \(p\)-value based on the \(t\)-reference distribution with \(df\)=84.}
\label{supp:stride}
\setlength{\tabcolsep}{4pt}
\renewcommand{\arraystretch}{1.12}

\resizebox{\textwidth}{!}{%
\begin{tabular}{
  l
  S[table-format=1.3] S[table-format=1.3] S[table-format=1.3]
  S[table-format=1.3] S[table-format=1.3] S[table-format=1.3]
  S[table-format=1.3] S[table-format=1.3] S[table-format=1.3]
  S[table-format=1.3] S[table-format=1.3] S[table-format=1.3]
}
\toprule
& \multicolumn{3}{c}{$W_D$}
& \multicolumn{3}{c}{$W_R$}
& \multicolumn{3}{c}{$W_O$}
& \multicolumn{3}{c}{DOOR} \\
\cmidrule(lr){2-4}\cmidrule(lr){5-7}\cmidrule(lr){8-10}\cmidrule(lr){11-13}
Method
& {Est.} & {SE} & {$p$-value}
& {Est.} & {SE} & {$p$-value}
& {Est.} & {SE} & {$p$-value}
& {Est.} & {SE} & {$p$-value} \\
\midrule

Wald test (clustered rank sum)
& 0.040 & 0.014 & 0.005
& 1.134 & 0.050 & 0.008
& 1.083 & 0.030 & 0.007
& 0.520 & 0.007 & 0.005 \\

Wald test (bivariate clustered U-statistics)
& 0.040 & 0.013 & 0.003
& 1.134 & 0.047 & 0.005
& 1.083 & 0.028 & 0.004
& 0.520 & 0.006 & 0.003 \\

Wald test (jackknife SE)
& 0.040 & 0.013 & 0.003
& 1.134 & 0.047 & 0.006
& 1.083 & 0.028 & 0.005
& 0.520 & 0.007 & 0.003 \\

\bottomrule
\end{tabular}%
}

\vspace{2mm}
\footnotesize
The three procedures are the Wald test based on the clustered rank-sum representation (FCL), the Wald test based on bivariate clustered U-statistics (ZJ), and the Wald test with delete-one-cluster jackknife standard errors. The reported \(p\)-values are based on the \(t\)-reference distribution.
\end{table}

\clearpage
\section{Tutorial for the \texttt{WinsCRT} package and example code} \label{sec:tutorial}

We developed \texttt{WinsCRT}, an \textsf{R} package for estimation and inference of win-statistics estimands in CRT using prioritized longitudinal event outcomes. The package is intended for CRT settings where each subject may contribute multiple time-stamped records, and event priority is encoded through an integer status variable. The package supports the win statistics 
\(
\text{WD} \ (\text{Win Difference}), 
\text{WR} \ (\text{Win Ratio}), 
\text{WO} \ (\text{Win Odds}),
\text{DOOR}.
\)
Here, DOOR is the probability-of-win estimand with ties split equally. Implemented inference methods are
\(\texttt{wald\_score},\ \texttt{wald\_u},\ \texttt{wald\_jk},\ \texttt{perm},\ \texttt{fs},\ \texttt{jel}.
\)
\texttt{WinsCRT} expects a long-format event-log data frame with one row per observed event/censoring record and at least the following columns:
\begin{itemize}
  \item \texttt{clu}: cluster identifier;
  \item \texttt{id}: subject identifier (within cluster);
  \item \texttt{trt}: cluster-level treatment indicator (\(0/1\));
  \item \texttt{t}: event/censoring time;
  \item \texttt{st}: status code.
\end{itemize}

Status coding convention:
\[
\texttt{st}=0 \text{ indicates censoring/no event}, \qquad
\texttt{st}\in\{1,2,\ldots\} \text{ indicates event types}.
\]
Larger positive status values correspond to higher priority (equivalently, smaller values correspond to lower priority). The main entry point is:
\begin{verbatim}
WinsCRT(data, cluster, subject, trt, time, status,
        method = c("wald_score","wald_u","wald_jk","perm","fs","jel"),
        estimand = c("WD","WR","WO","DOOR"),
        null = NULL, alternative = c("two.sided","greater","less"),
        alpha = 0.05, use_t = TRUE, B = 2000, seed = NULL,
        keep = NULL, strict = TRUE)
\end{verbatim}

The function returns an object of class \texttt{"WinsCRT"} containing:
\begin{itemize}
  \item point estimate,
  \item standard error (if available),
  \item test statistic,
  \item \(p\)-value,
  \item confidence interval (if available),
  \item method-specific details.
\end{itemize}
Associated \texttt{print()} and \texttt{summary()} methods are implemented. The package includes an example dataset \texttt{dat}. A minimal usage example is:

\begin{verbatim}
library(WinsCRT)

data(dat)

fit <- WinsCRT(
  data     = dat,
  cluster  = "clu",
  subject  = "id",
  trt      = "trt",
  time     = "t",
  status   = "st",
  method   = "wald_score",
  estimand = "WD"
)

print(fit)
summary(fit)
\end{verbatim}

To compare methods for the same estimand:

\begin{verbatim}
methods <- c("wald_score","wald_u","wald_jk","perm","fs","jel")

res <- lapply(methods, function(m) {
  z <- WinsCRT(
    data     = dat,
    cluster  = "clu",
    subject  = "id",
    trt      = "trt",
    time     = "t",
    status   = "st",
    method   = m,
    estimand = "WD",
    B        = 1000,
    seed     = 123
  )
  data.frame(method = m, estimate = z$estimate, p_value = z$p_value)
})

do.call(rbind, res)
\end{verbatim}

For ratio estimands (\texttt{WR}, \texttt{WO}), inference is performed on the log scale where implemented, while point estimates and confidence limits are reported on the original scale.

\end{document}